\documentclass[letterpaper,11.5pt]{article}
\usepackage[latin1]{inputenc} 
\usepackage[top=3cm, bottom=3cm, left=3cm, right=3cm]{geometry}
\usepackage{graphicx,rotating,amsthm,amsbsy,amssymb,amsmath, caption, natbib, bbm, color, thmtools}
\usepackage[ruled,vlined]{algorithm2e}

\usepackage{amsmath, amssymb, bbm, xcolor}

\newtheorem{defi}{Definition}
\newtheorem{prop}{Proposition}

\newcommand{\E}{\mathbbm{E}}

\newtheorem{definition}{Definition}
\newtheorem{proposition}{Proposition}
\newtheorem{example}{Example}

\usepackage{xr}

\makeatletter
\newcommand*{\addFileDependency}[1]{
  \typeout{(#1)}
  \@addtofilelist{#1}
  \IfFileExists{#1}{}{\typeout{No file #1.}}
}
\makeatother

\newcommand*{\myexternaldocument}[1]{%
    \externaldocument{#1}%
    \addFileDependency{#1.tex}%
    \addFileDependency{#1.aux}%
}
\myexternaldocument{SupplementaryMaterial}
\myexternaldocument{main}

\title{Differentially private methods for \\ managing model uncertainty in  linear regression models}

\begin{document}

\begin{center}
    {\Large \textbf{Differentially private methods for managing \\  model uncertainty in  linear regression models}}  \\
     V\'ictor Pe\~na\footnote{Corresponding author: \texttt{victor.pena.pizarro@upc.edu}}, Universitat Polit\`ecnica de Catalunya \\
     Andr\'es Felipe Barrientos, Florida State University
\end{center}

\begin{abstract}
In this work, we propose differentially private methods for hypothesis testing, model averaging, and model selection for normal linear models. We consider Bayesian methods based on mixtures of $g$-priors and non-Bayesian methods based on likelihood-ratio statistics and information criteria. The procedures are asymptotically consistent and straightforward to implement with existing software. We focus on practical issues such as adjusting critical values so that hypothesis tests have adequate type I error rates and quantifying the uncertainty introduced by the privacy-ensuring mechanisms.  
\end{abstract}

\section{Introduction}

Differential privacy \citep{dwork2006calibrating} is a formal framework for quantifying the privacy of randomized algorithms. Its theoretical properties have been studied extensively (see e.g. \cite{dwork2014algorithmic}) and it has been adopted by companies such as Google, Apple, and Microsoft \citep{garfinkel2018issues} as well as institutions like the U.S. Census Bureau \citep{abowd2018us}. 

In this article, we develop differentially private methods for normal linear models. We propose differentially private hypothesis tests for comparing nested models (in Section~\ref{sec:DPBF}) as well as methods for model averaging and selection (in Section~\ref{sec:modunc}). We consider Bayesian methods based on mixtures of $g$-priors \citep{liang2008mixtures} and non-Bayesian methods that are built upon likelihood-ratio statistics and information criteria.

{In Bayesian hypothesis testing and model selection, prior distributions must be chosen carefully because their effect does not vanish as the sample size grows \citep{bayarri2012criteria}. Our work is based on mixtures of $g$-priors because, when combined with right-Haar priors on the common parameters, they satisfy a list of appealing criteria proposed in \cite{bayarri2012criteria} and they are conveniently implemented in the \texttt{R} package  \texttt{library(BAS)} \citep{BAS}.}

\textcolor{black}{From a non-Bayesian perspective, we work with likelihood-ratio tests and information criteria. We make this choice because of their theoretical properties, intuitive appeal, and ease of use. The class of information criteria we consider includes the Akaike Information Criterion (AIC; \cite{akaike1974new}) and the Bayesian Information Criterion (BIC; \cite{schwarz1978estimating}), among others. Information criteria have been proven useful in model averaging and selection problems for the normal linear model and beyond. We recommend the monograph \cite{claeskens2008model} for an extensive overview of the approach.}

We enforce differential privacy with well-established techniques. For hypothesis testing, we use the subsample and aggregate technique \citep{nissim2007smooth, smith2011privacy}, which consists in splitting the data into subgroups and releasing perturbed averages. \textcolor{black}{For model averaging and selection, we use \textit{sufficient-statistic perturbation}, which consists in releasing a noisy version of a sufficient statistic (see, for example, \cite{mcsherry2009differentially, vu2009differential} and \cite{bernstein2019differentially}).}

\color{black}

\subsection{Related Work}

There is a growing literature on differentially private methods for linear regression models. For example, \cite{amitai2018differentially} estimate quantiles and posterior tail probabilities of coefficients, \cite{barrientos2019differentially} test the significance of individual regression coefficients, and \cite{ferrando2022parametric} propose methods for point and interval estimation that can be applied to the normal linear model. \cite{lei2018differentially} consider the problem of model selection based on information criteria, but do not consider model averaging or Bayesian approaches, and \cite{bernstein2019differentially} propose a method for sampling from posterior distributions on regression coefficients, but do not consider hypothesis testing, model averaging or selection. 

Our methods for hypothesis testing involve data-splitting and censoring. Both operations can induce bias in the outputs. \cite{evans2020statistically} considers the effects of censoring in estimates that are asymptotically normal and proposes strategies to correct the bias induced by censoring. \cite{covington2021unbiased} uses bags of little bootstraps \citep{kleiner2012big} to provide unbiased estimates and valid confidence intervals. On the other hand, \cite{ferrando2022parametric} use the parametric bootstrap for bias correction. In hypothesis testing, the bias induced by data-splitting and censoring leads to inappropriate critical values for rejecting null hypotheses. We address this issue by simulating the distribution of the differentially private test statistics under the null hypothesis and find corrected critical values that are adequately calibrated. 

In general, our Bayesian methodology draws from the objective Bayesian literature for hypothesis testing, averaging, and selection, especially from \cite{liang2008mixtures} and \cite{bayarri2012criteria}. In these references, the classes of priors we consider here are proposed after showing that they satisfy a list of conceptually appealing criteria.
 
\subsection{Main Contributions}

In Section~\ref{sec:DPBF}, we argue that working on a logarithmic scale is natural for defining differentially private Bayes factors. We show that the methods are asymptotically consistent under regularity conditions that are similar to the ones needed for consistency when there are no privacy constraints. In Section~\ref{subsec:quant}, we provide a simple strategy to quantify the effects of the privacy-ensuring mechanisms. Then, in Section~\ref{subsec:cens}, we study the effects of censoring and data-splitting in our methodology.

In Section~\ref{sec:modunc}, we use sufficient-statistic perturbation to define methods for model averaging and selection. If we took a naive approach to the problem, the variance of the perturbation term would be exponential in the number of predictors. Sufficient-statistic perturbation allows us to release a statistic such that the variance of the perturbation term is quadratic in the number of predictors. The methods for model selection  are consistent under conditions that are similar to those needed for consistency without privacy constraints. In Section~\ref{subsec:hist}, we propose a strategy that quantifies the uncertainty introduced by the mechanisms that is analogous to the one pursued in Section~\ref{subsec:quant}.

From a practical point of view, we give guidelines for maximizing the statistical utility of the methods in finite samples. We illustrate the performance of our methods in Sections~\ref{subsec:hsb2} and~\ref{subsec:empeval}. We include additional results from the simulation studies in the Appendix.

The proofs to the Propositions stated in the main text are relegated to Appendix~\ref{suppl:section:Proofs}.

\color{black}

\subsection{Notation}

Extrema appear often in Section~\ref{sec:DPBF} because the methods are based on censored statistics. Our notation for them is $a \vee b = \max(a,b)$ and $a \wedge b = \min(a, b)$. The censored statistics are of the form $T^c = (T \vee a) \wedge b$: that is, $T^c = a$ whenever $T \le a$ and $T^c = b$ whenever $T \ge b$ (with the understanding that $a < b$).

We use the following notation for probability distributions. The $p$-variate normal distribution with mean $\mu$ and covariance matrix $\Sigma$ is $N_p(\mu, \Sigma)$, the Laplace distribution with location parameter $\mu$ and scale parameter $b$ is $L_1(\mu, b)$, and the  $(p \times p)$-dimensional Wishart distribution with degree of freedom $n > p-1$ and positive-definite scale matrix $V$ is  $W_{p}(n, V)$. 

In the case of matrices, all of them are assumed to be full-rank unless otherwise stated. Our notation for basic matrix operations and special matrices is as follows. The matrix transpose of $A$ is $A'$, the perpendicular projection operator onto the column space of $A$ is $P_A = A(A'A)^{-1} A$, and the upper Cholesky factor of $A$ is $A^{1/2}$. The $(n \times n)-$dimensional zero matrix is $0_{n \times n}$ and the $(n \times n)$-dimensional identity matrix is $I_n$. For vectors, the usual $p$-norm on $\mathbbm{R}^d$ is $\lVert \cdot \rVert_p$. The zero vector is $0_n = (0, 0, \, ... \, {}_{n)} , 0)'$ and the vector of ones is $1_n = (1, 1, \, ... \, {}_{n)} , 1)'$. The expectation of a random variable $X$ is  $\mathbb{E}(X)$.

\section{Brief Review of Differential Privacy} \label{sec:review}

Differential privacy \citep{dwork2006calibrating,dwork2014algorithmic} is a probabilistic property of randomized algorithms. Intuitively, differential privacy limits how much we can learn about individual entries in a data set. 

In the literature, randomized algorithms that ensure differential privacy are referred to as \textit{mechanisms}. Conceptually, mechanisms are functions $\mathcal{M}$ that take data $D$ as inputs and output a random $\mathcal{M}(D)$ that is, in some sense, private. In the definition of differential privacy, a key notion is that of \textit{neighboring datasets}. Two datasets $D$ and $\tilde{D}$ are neighbors, which is denoted by $D \sim \tilde{D}$, if they only differ in one row.

\begin{defi} 
[$(\varepsilon, \delta)$-differential privacy]
\label{def:DiffPriv} 
Let $\varepsilon > 0$ and $0 \le \delta \le 1$. A mechanism $\mathcal{M}$ satisfies $(\varepsilon,\delta)$-differential privacy if, for all $D \sim \tilde{D}$ any $\mathcal{M}$-measurable set $S$,
we have $P[\mathcal{M}(D) \in S] \leq \exp(\varepsilon) P[\mathcal{M}(\tilde{D}) \in S]+\delta.$
\end{defi}

When $\delta=0$, we retrieve $\varepsilon$-differential privacy, which is the most popular formalization of privacy in the literature. When $\varepsilon$ is small, the privacy of $\mathcal{M}$ increases; in such cases, the probability distributions of $\mathcal{M}(D)$ and $\mathcal{M}(\tilde{D})$ are forced to be similar, so the output of $\mathcal{M}$ is not very sensitive to small changes in $D$. On the other hand, as $\varepsilon$ increases, the distributions of $\mathcal{M}(D)$ and $\mathcal{M}(\tilde{D})$ are allowed to be more different, which decreases the privacy of the output.  When $0 < \delta \le 1$, Definition~\ref{def:DiffPriv} allows $\mathcal{M}(D)$ to release outputs that lead to a ``high privacy loss'' with probability $\delta$ (see Section 2 in \cite{dwork2014algorithmic} for details).

In this article, $\mathcal{M}(D)$ are perturbed versions of confidential statistics $T(D)$. The scale of the perturbation depends on the global sensitivity of $T(D)$, a concept we define below. 

\begin{defi}[Global sensitivity]
Let $T$ be a statistic mapping data to $\mathbbm{R}^d$. The global sensitivity of $T$ is defined as $\Delta_p = \sup_{D \sim \tilde{D}} \lVert T(D) - T(\tilde{D}) \rVert_{p}$.
\end{defi}

A key property of differential privacy is that transformations of $(\varepsilon, \delta)$-differentially private statistics are $(\varepsilon, \delta)$-differentially private. In the literature, this is known as the \textit{post-processing} property of differential privacy. We use this property frequently; for example, we use it to find posterior probabilities of hypotheses given Bayes factors.

\color{black}
\section{Brief Review of Bayes Factors and Information Criteria}  \label{sec:review2}

In this section, we review facts of Bayes factors and information criteria that are relevant for our purposes. This is not meant to be a comprehensive review; we refer the reader to \cite{berger2001objective}, \cite{liang2008mixtures}, and \cite{bayarri2012criteria} for further background on Bayesian methods and \cite{claeskens2008model} for information criteria.

\subsection{Hypothesis Testing}

In Section~\ref{sec:DPBF}, we introduce differentially private methods for hypothesis testing. We cover Bayesian approaches based on Bayes factors and non-Bayesian approaches based on likelihood-ratio tests and information criteria. Now, we review some core concepts in Bayesian and non-Bayesian testing that are helpful for contextualizing our work. 

Let $y = (y_1, \, ... \, , y_n)$ be a vector collecting independent and identically distributed observations from a statistical model with sampling density $f(y \mid \theta) = \prod_{i=1}^n f(y_i \mid \theta)$, for $\theta \in \Theta$. Our goal is testing $H_0: \theta \in \Theta_0$ against $H_1: \theta \in \Theta_1$, where $\Theta_0$ and $\Theta_1$ are disjoint subsets of $\Theta$.

In the Bayesian paradigm, all unknowns have probability distributions associated to them, including the hypotheses $H_0$ and $H_1$ and the parameter $\theta$. Before observing any data, the uncertainty in the hypotheses is reflected in the prior probabilities $P(H_0)$ and $P(H_1) = 1 - P(H_0)$. The uncertainty in $\theta$ is usually expressed conditionally through the prior distributions $\pi(\theta \mid H_0)$ and $\pi(\theta \mid H_1)$. Upon observing data, the uncertainty about $\theta$ and the hypotheses is updated in the posterior distribution, which is simply the conditional distribution of these unknowns given the data. 

A key quantity in Bayesian hypothesis testing is the Bayes factor, which we define now. Let $o_{10} = P(H_1)/P(H_0)$ be the prior odds of the hypotheses. Then, the Bayes factor of $H_1$ to $H_0$ is defined as
$$
B_{10} = {\int_{\theta \in \Theta_1} \pi( \theta \mid H_1) f(y \mid \theta)  \nu(\mathrm{d} \theta) }/{\int_{\theta \in \Theta_0} \pi( \theta \mid H_0) f(y \mid \theta)  \, \nu(\mathrm{d} \theta)}
$$
for a dominating measure $\nu(\cdot)$.

Bayes factors are important in Bayesian testing for several reasons. For example, the posterior probability of $H_1$ is
$$
P(H_1 \mid y) = {o_{10} B_{10}}/(1+ o_{10} B_{10}),
$$
which depends on the data only through the Bayes factor $B_{10}$. Bayes factors can be motivated as ratios of integrated likelihoods that quantify the evidence in favor of $H_1$ relative to $H_0$ \citep{berger1999integrated}, and there have been efforts to categorize the strength of evidence in favor or against $H_1$ based on the magnitude of $B_{10}$. Perhaps the most popular categorization is ``Jeffreys' scale of evidence'' (\cite{jeffreys1939}, see Table~\ref{tab:jeffreys}).

\begin{table}[ht]
\centering
\caption{Jeffreys' scale of evidence for Bayes factors \citep{jeffreys1939}}
\label{tab:jeffreys}
\begin{tabular}{ll}
\hline
Bayes factor &  Interpretation                  \\ \hline
$B_{10} < 1$                                & $H_0$ supported                                                \\
$1 < B_{10} < 10^{1/2}$                     & Evidence against $H_0$, but not worth more than a bare mention \\
$10^{1/2} < B_{10} < 10$                    & Evidence against $H_0$ substantial                             \\
$10 < B_{10} < 10^{3/2}$                    & Evidence against $H_0$ strong                                  \\
$10^{3/2} < B_{10} < 10^2$                  & Evidence against $H_0$ very strong                             \\
$B_{10} > 10^2$                             & Evidence against $H_0$ decisive                               
\end{tabular}
\end{table}

Logarithms of Bayes factors are naturally symmetric about zero. To see this, assume that $H_0$ and $H_1$ are equally likely \textit{a priori}. Then, the posterior probability of $H_1$ is the standard logistic function in $\log B_{10}$: that is, $P(H_1 \mid D) =  1/[1+\exp(- \log B_{10})]$. Changing the sign of $\log B_{10}$ leads to the complement $1-P(H_1 \mid D)$, and $P(H_1 \mid D) = 1/2$ if and only if $\log B_{10} = 0$. 

In Bayesian hypothesis testing, the priors  $\pi(\theta \mid H_0)$ and $\pi(\theta \mid H_1)$ must be chosen carefully. Vague priors on $\theta$, which are commonplace in estimation problems, can lead to posterior probabilities that overwhelmingly support $H_0$ no matter what the data are. This phenomenon is known in the literature as Lindley's paradox (see  \cite{lindley1957statistical} and \cite{robert2014jeffreys}). In normal linear regression problems, mixtures of $g$-priors have been studied carefully in \cite{liang2008mixtures} and \cite{bayarri2012criteria}. They have desirable properties such as invariance of Bayes factors with respect to changes of measurement units and large-sample consistency. In Sections~\ref{sec:DPBF} and \ref{sec:modunc}, we work with priors within this class.

From a non-Bayesian perspective, we propose using likelihood-ratio tests and information criteria. The likelihood ratio of $H_1$ to $H_0$ is defined as

$$
\Lambda_{10} = {\max_{\theta \in \Theta_1} f(y \mid \theta)}/{\max_{\theta \in \Theta_0} f(y \mid \theta)}.
$$

The likelihood ratio $\Lambda_{10}$ is similar to the Bayes factor $B_{10}$, but instead of averaging the likelihoods with weights given by $\pi(\theta \mid H_1)$ and $\pi(\theta \mid H_0)$, the likelihoods are maximized under $H_1$ and $H_0$. Likelihood ratio tests are standard within the field of statistics and their properties are well-characterized (see, for example, \cite{lehmann2006testing}). Under mild conditions, the transformed log-likelihood ratio $2 \log \Lambda_{10}$ is asymptotically distributed as chi-squared under $H_0$ and consistent under $H_1$. 

Finally, we define the class of information criteria
$$
I_{10} = n^{-\rho/2} \, \Lambda_{10},
$$
which encompasses AIC for $\rho = 2 p / \log n$, BIC for $\rho = p$, and the likelihood ratio statistic $\Lambda_{10}$ for $\rho = 0$. 

For a fixed $\rho$, the information criterion $I_{10}$ can be interpreted as a penalized likelihood ratio, where $\rho$ acts as a penalty for model complexity. For the normal linear model, likelihood-ratio tests and AIC fail to be consistent under $H_0$; BIC, on the other hand, is consistent under $H_0$. We refer the reader to Chapter 4 in the monograph \cite{claeskens2008model} for a more general version of this result and a discussion on related issues.

\subsection{Model Averaging and Selection}

In Section~\ref{sec:modunc}, we focus on model averaging and selection. The context is a regression problem where there is an outcome variable $Y \in \mathbb{R}^{n}$ and $p$ predictors collected in a design matrix $X \in \mathbb{R}^{n \times p}$. We do not know which variables in $X$, if any, should be included in our model for $Y$ given $X$. This type of uncertainty is often referred to as \textit{model uncertainty}. For an introduction on the topic with a strong Bayesian flavor, we recommend \cite{draper1995assessment}.

Model uncertainty can be parameterized through a binary vector $\gamma \in \{0,1\}^p$ that indicates active predictors: $\gamma_i = 0$ if the $i$th predictor is not active and $\gamma_i = 1$ if it is. Conceptually, we assume that there is a true model generating the data identified by $T \in \{0,1\}^p$. Given finite data $D$, we have uncertainty as to what $T$ is.  From a Bayesian perspective, we can put a prior on $\gamma$ and find posterior probabilities to quantify this uncertainty; from a non-Bayesian perspective, we can find a point estimate of $\gamma$ or average our uncertainty over it with rules inspired by Bayesian procedures. 

For each model, which we identify by its active predictors in $\gamma \in \{0, 1\}^p$, we can compute a Bayes factor or information criterion relative to the null model, which does not include any active predictors. We denote these null-based Bayes factors and information criteria $B_{\gamma0}$ and $I_{\gamma0}$, respectively. With these, we can perform model selection after maximizing over $\gamma$ or we can find model-averaged estimates with weights proportional to $B_{\gamma0}$ or $I_{\gamma 0}$.

\color{black}

\section{Hypothesis Testing} \label{sec:DPBF}

In this section, we describe differentially private methods for testing a null hypothesis $H_0$ against an alternative $H_1$. The methodology described here can be applied in general, but in Section~\ref{subsec:lm} we focus on nested linear regression models, for which we have derived theoretical guarantees. 

Our methods are based on the subsample and aggregate technique \citep{nissim2007smooth}. It consists in splitting the data into $M$ disjoint subgroups, computing statistics within the subgroups, and averaging the results. The output is made differentially private by adding a perturbation term $\eta$ whose variance is increasing in the global sensitivity of the statistics involved. 

We apply the subsample and aggregate technique as follows. First, we split the data into $M$ disjoint subgroups of sample size $b_1, b_2, \, ... \, , b_M$ ($\sum_{i=1}^M b_i = n$) and compute censored statistics $T_i^c = (T_i \vee L) \wedge U \in [L, U]$ for $i \in \{1, 2, \, ... \, , M\}$, where $T_i$ are raw statistics computed from confidential data. After censoring the $T_i$, we know that they have global sensitivity $\Delta = U - L$.  Finally, we release the noisy average $\sum_{i=1}^M T_i^c/M + \eta$, where $\eta$ is a random perturbation term that ensures $(\varepsilon, \delta)$-differential privacy. If $\delta =0$, then $\eta \sim L_1(0, \Delta/(M \varepsilon))$; if $0 < \delta \le 1$, then $\eta \sim {N}_1(0, a \Delta^2/(2 M \varepsilon))$, where $a$ is a constant that can be computed with Algorithm 1 in \cite{balle2018improving}.  

The confidential statistics $T_i$ we work with are logarithms of Bayes factors and logarithms of information criteria. We justify working on a logarithmic scale for Bayes factors in the subsequent paragraphs. A similar argument can be used to justify working with logarithms of information criteria as well.

Let $B_{10,i}^c$ be the censored Bayes factor in the $i$th subgroup. Naively, one could release the noisy average $\sum_{i=1}^M B_{10,i}^c/M + \eta.$ Unfortunately, that approach has undesirable properties. Under both the Laplace and analytic Gaussian mechanisms, $\eta$ is supported in $\mathbbm{R}$ and symmetric around zero, whereas $B_{10}$ is always non-negative and shows equal support to $H_0$ and $H_1$ when $B_{10} = 1$.  If there are no privacy constraints, Bayes factors satisfy $B_{01} = B_{10}^{-1}$, but in general $\sum_{i=1}^M B^c_{01,i}/M + \eta  \neq ( \sum_{i=1}^M B^c_{10,i}/M + \eta)^{-1}$. 

Alternatively, we propose working on a logarithmic scale, defining
$$
\log \tilde{B}_{10} = \sum_{i=1}^M \log B^c_{10,i}/M + \eta, \, \, \, \log B^c_{10, i} = (\log B_{10, i} \vee L) \wedge U.
$$
Logarithms of Bayes factors are supported on $\mathbbm{R}$ and, as we argued in Section~\ref{sec:review2}, they have a natural symmetry around zero, so it is sensible to add a zero mean perturbation term on that scale. After exponentiating, we obtain a geometric mean of censored Bayes factors with a multiplicative perturbation:

$$
\tilde{B}_{10} = \exp(\eta) \left( \prod_{i=1}^M B^c_{10,i}\right)^{1/M}.
$$

Since the distribution of $\eta$ is symmetric, $(\tilde{B}_{10})^{-1}$ is equal in distribution to $\tilde{B}_{01}$, and it is exactly equal to $\tilde{B}_{01}$ when $\eta = 0$. Geometric means of Bayes factors have appeared in the objective Bayesian literature in geometric intrinsic Bayes factors \citep{berger1996intrinsic}.  

The statistic $\tilde{B}_{10}$ is based on censored Bayes factors that are supported on $[L, U]$. However, the support of $\tilde{B}_{10}$ is not $[L, U]$ after introducing the perturbation term $\eta$.  This issue can be solved by censoring $\tilde{B}_{10}$, defining
\begin{equation} \label{eq:bstar}
B^\ast_{10} = (\tilde{B}_{10}
\vee L^\ast) \wedge U^\ast,
\end{equation}
where $L^\ast = \exp (L)$ and $U^\ast = \exp (U)$. With $B^\ast_{10}$, we can define  the differentially private posterior probability of $H_1$ given $D$ as $P^\ast(H_1 \mid D) = [1-P(H_0)] B^{\ast}_{10}/\{P(H_0) + [1-P(H_0)] B^\ast_{10}\}$. 

Following the same reasoning, we can define a differentially private information criterion 
\begin{align} \label{eq:istar}
I^\ast_{10} &= (\tilde{I}_{10} \vee L^\ast) \wedge U^\ast \\
\tilde{I}_{10} &= \exp(\eta) \left( \prod_{i=1}^M I_{10, i}^c \right)^{1/M} \nonumber \\ 
I_{10, i}^c &= ( b_i^{-\rho/2} \Lambda_{10, i} \vee  L ) \wedge U \nonumber.
\end{align}

\subsection{Nested Linear Regression Models} \label{subsec:lm}

Consider the normal linear model
$$
Y = X_0 \beta_0 + X \beta + \sigma W, \, \, \, W \sim N_n(0_n, I_n),
$$ 
where $X_0 \in \mathbbm{R}^{n \times p_0}$ and $X \in \mathbbm{R}^{n \times p}$ are full-rank and $n > p + p_0$. In this section, we present differentially private methods for testing $H_0: \beta = 0_p$ against $H_1: \beta \neq 0_p$. The set of predictors in $X_0$ is common to $H_0$ and $H_1$ and it can be, for instance, an intercept $1_n$. 

For the Bayesian methods, we define our priors after reparameterizing the model. We rewrite the model as $Y = X_0 \psi + V \beta + \sigma W$ for $V = (I_n - P_{X_0}) X$ and $\psi = \beta_0 + (X_0'X_0)^{-1} X_0'X \beta$. In this parameterization, the common predictors $X_0$ are orthogonal to $V$, which is specific to $H_1$. If $X_0$ is an intercept $1_n$, the reparameterization simply centers the predictors in $X$.

Our prior specification is
\begin{align*}
    \pi(\psi, \sigma^2) \propto 1/\sigma^2, \, \, \, 
    \pi(\beta \mid \sigma^2, H_1) = \int_{0}^\infty N_p(\beta \mid 0_p, g \sigma^2 (V'V)^{-1}) \, \pi(g) \, \nu( \mathrm{d} g).
\end{align*}
where $\nu$ is an appropriate dominating measure. We allow the prior measure on $g$ to depend on $n$ and $p$, but not on $Y$.

The prior on $(\psi, \sigma^2)$ is the right-Haar prior for this problem, which is improper, whereas $\beta \mid \sigma^2, H_1$ is a mixture of $g$-priors. This class of priors has strong theoretical support (see \cite{liang2008mixtures} and \cite{bayarri2012criteria} for details). For example, it leads to Bayes factors and posterior probabilities of hypotheses that are invariant with respect to invertible linear transformations of the design matrix $V$, such as changes of units.  This property would not be satisfied if we had chosen a diagonal covariance matrix for $\beta \mid g, \sigma^2, H_1$.   

The prior distribution for $\beta_0$ and $\sigma^2$ is improper, so the marginal distributions can only be defined up to arbitrary multiplicative constants. We use the same constants for both $H_0$ and $H_1$, a decision that is supported by the principle of invariance and other justifications discussed in \cite{berger1998bayes} and \cite{bayarri2012criteria}.

When there data are not confidential, we can report the Bayes factor \citep{liang2008mixtures}: 
\begin{equation} \label{eq:bf}
B_{10} = \int_0^\infty (g+1)^{(n-p-p_0)/2} [1 + g(1 - R^2)]^{-(n-p_0)/2} \, \pi(g) \,  \nu(\mathrm{d} g),
\end{equation}
where $R^2 = Y' P_V Y/Y'(I_n -  P_{X_0})Y$ or, from a non-Bayesian perspective, the information criterion
\begin{equation} \label{eq:info}
    I_{10} = n^{-\rho/2} \Lambda_{10} = n^{-\rho/2} (1-R^2)^{-n/2}.
\end{equation}

If there are privacy constraints, $B_{10}$ or $I_{10}$ cannot be released directly. We propose applying the subsample and aggregate technique to release differentially private versions of $B_{10}$ and $I_{10}$. That is, we split up the data set into $M$ disjoint subgroups of sample sizes $b_1, b_2, \, ... \, , b_M$ ($\sum_{i=1}^M b_i = n$), define priors $\pi_i(g_i)$ and penalties $\rho_i$ for $i \in \{1,2, \, ... \, , M\}$, and release $B^{\ast}_{10}$ and $I^{\ast}_{10}$ as defined in Equations~\eqref{eq:bstar} and~\eqref{eq:istar}, respectively.

Proposition~\ref{prop:consDPBF} below states that $B^{\ast}_{10}$  and $I^{\ast}_{10}$ are consistent under certain regularity conditions. Consistency does not follow from \cite{liang2008mixtures} for various reasons. Firstly, there is a growing number of subgroups and the sample sizes within the subgroups grow to infinity. Another difference is that the Bayes factors are censored and there is a perturbation term $\eta$ that is introduced to ensure differential privacy. Our result does not follow directly from \cite{smith2011privacy}, either. \cite{smith2011privacy} studies the asymptotic behavior of averages $\sum_{i=1}^M T_i/M$ released with the subsample and aggregate method. If $T^\ast_i$ is the ``true value'' of $T_i$, Lemma 6 in \cite{smith2011privacy} assumes that $\sqrt{b_i} (T_i - T^\ast_i)$ is asymptotically normal, $\mathbb{E}(T_i) - T^\ast_i \in O(1/b_i)$, and $\mathbb{E}\{[\sqrt{b_i} (T_i - T^\ast_i)]^3\} \in O(1).$ Furthermore, the results in \cite{smith2011privacy} are for independent and identically distributed data. In our case, it is unclear whether the asymptotic conditions are satisfied (if they are, it would require proof), and we would need additional assumptions on, at least, the design matrices, the priors $\pi_i(g_i)$, and the penalties $\rho_i$. Instead of verifying the conditions in \cite{smith2011privacy}, our proof uses a union bound and tail inequalities for $R^2$.

\begin{prop} \label{prop:consDPBF} Under the regularity conditions listed below, $B^{\ast}_{10}$ and $I^\ast_{10}$ are consistent: $B^\ast_{10}\rightarrow_P 0$ and $I^{\ast}_{10} \rightarrow_P 0$ when $H_0$ is true and $(B^{\ast}_{10})^{-1} \rightarrow_P 0$ and $(I^{\ast}_{10})^{-1} \rightarrow_P 0$ when $H_1$ is true.
\begin{enumerate}
    \item \textbf{Well-specified model:} The raw confidential data are generated from the normal linear model described in this section.
    \item \textbf{Growth of $M$ and $b_i$:} $\lim_{n \rightarrow \infty} M = \infty$ and $\lim_{n \rightarrow \infty} \inf_{i \in 1:M} b_i = \infty$. Moreover, if $p \ge 2$,   $\lim_{n \rightarrow \infty} \sup_{i \in 1:M} M/ b_i = 0$ and, if $p = 1$,  $\lim_{n \rightarrow \infty} \sup_{i \in 1:M} M \sqrt{\log b_i/ b_i} = 0$.
    \item  \textbf{Censoring limits:} $\lim_{n \rightarrow \infty} L =  - \infty$, $\lim_{n \rightarrow \infty} U = \infty$.
    \item \textbf{Privacy parameters:} The privacy parameters $\varepsilon$ and $\delta$ are such that $\lim_{n \rightarrow \infty} (U-L)/(M \varepsilon) = \lim_{n \rightarrow \infty} a(U-L)^2/(M \varepsilon) = 0$. 
    \item \textbf{Design matrices:} Under $H_1$,  $\lim_{n \rightarrow \infty} \inf_{i \in 1:M} \beta_T' X_i' X_i \beta_T/(\sigma^2_T b_i)  > 0$, where $\beta_T$ and $\sigma^2_T$ are the fixed true values of $\beta$ and $\sigma^2$, respectively. 
    \item \textbf{Priors on $g_i$ and penalties $\rho_i$:}  
\begin{align*}
 &\lim_{n \rightarrow \infty} \sup_{i \in 1:M} \int_{0}^\infty  b_i^{p/2} \, (g_i+1)^{-p/2} \pi_i(g_i) \, \nu( \mathrm{d} g_i ) < \infty, \, \, \rho_i \le \max(p,  \log b_i)  \\
&\lim_{n \rightarrow\infty} \inf_{i \in 1:M}  \int_{b_i}^\infty  b_i^{p/2}  (g_i+1)^{-p/2}  \pi_{i}(g_i) \, \nu(\mathrm{d} g_i) > 0, \, \,  \rho_i \ge p .
\end{align*}
\end{enumerate}
\end{prop}

Assumption 1 states that the model for the confidential data is well-specified. We explicitly state this assumption because it is not necessarily satisfied in Section~\ref{sec:modunc}. Assumption 2 implies that both the number of subgroups $M$ and the sample sizes within the subgroups $b_i$ grow to infinity, and it imposes conditions on the growth of $M$ relative to the growth of $b_i$. Assumption 3 assumes that the censoring limits increase with the sample size. Assumption 4 is related to the privacy parameters $\varepsilon$ and $\delta$. It holds for fixed $\varepsilon$ and $\delta$, but it also includes cases where $\varepsilon$ and $\delta$ depend on $n$. However, the dependency has to be such that, ultimately, $\eta$ converges to zero in probability: the privacy parameters cannot decrease too fast, relative to the censoring limits.  Assumption 5 is a regularity condition on design matrices that is commonly made in the literature (see Equation (22) in \cite{liang2008mixtures} or Equation (17) in \cite{bayarri2012criteria}). Finally, when combined with the previous assumptions, Assumption 6 is sufficient for establishing consistency. It imposes conditions on the prior on $g_i$ or, from a non-Bayesian perspective, the model complexity penalties $\rho_i$. Priors such as Zellner's $g$-prior with $g_i = b_i$, the robust prior proposed in \cite{bayarri2012criteria}, and Zellner-Siow \citep{zellner1980posterior} satisfy the conditions. The differentially private version of BIC satisfies the conditions on $\rho_i$, but AIC and the likelihood ratio statistic do not. When there are no privacy constraints, both AIC and the likelihood ratio statistic fail to be consistent under $H_0$ \citep{claeskens2008model}.  

Data-splitting is a commonly used strategy for analyzing big data (see, for instance, \citet{chen2021divide} for an up-to-date review on the topic). Proposition~\ref{prop:consDPBF} establishes consistency of mixtures of $g$-priors and information criteria if the data are split into independent subgroups, with or without privacy constraints (the latter corresponds to the case where $\varepsilon \rightarrow \infty$). To the best of our knowledge, there are no equivalent results in the literature. 

The Bayes factors and information criteria we work with are misspecified, since the observed values are perturbed averages of censored statistics. Therefore, Proposition~\ref{prop:consDPBF} can be of interest to those studying the properties of Bayesian hypothesis testing and selection in misspecified models, in the spirit of \cite{rossell2019additive}.

\subsection{Quantifying the Uncertainty Introduced by the Mechanism} \label{subsec:quant}

Given a differentially private Bayes factor, likelihood ratio or information criterion, we can find a confidence interval that quantifies the uncertainty introduced by the privacy-ensuring mechanism.

Let $\tilde{T} = \sum_{i=1}^n T_i^c/M + \eta$ be a differentially private statistic, and let $\eta_{1-\alpha/2}$ be the $(1-\alpha/2)$-th quantile of the perturbation term $\eta$. Then, $\mathcal{\tilde{C}}_\alpha = \tilde{T} \pm \eta_{1-\alpha/2} = [\tilde{l}_{\alpha}, \tilde{u}_{\alpha}]$ is a $1-\alpha$ confidence interval for $\sum_{i=1}^n T_i^c/M$. The interval can be shortened because $ \sum_{i=1}^n T_i^c/M \in [L, U]$, so $\mathcal{C}_{1-\alpha} = [ (\tilde{l}_{\alpha} \vee L) \wedge U, (\tilde{u}_{\alpha} \vee L) \wedge U]$ is also a $1-\alpha$ confidence interval for $\sum_{i=1}^n T_i^c/M$. 

Confidence intervals for transformations $f(\tilde{T})$ can be found by applying $f$ to the the endpoints of $\mathcal{C}_{1-\alpha}$. For Bayes factors of $H_1$ to $H_0$, the transformation is $f(x) = \exp(x)$; for posterior probabilities of $H_1$, it is $f(x) = [1-P(H_0)] \exp(x)/\{ P(H_0) + [1-P(H_0)] \exp(x)\}$.

\color{black}
\subsection{Effects of Censoring and Data-splitting} \label{subsec:cens}

The subsample and aggregate technique requires the specification of censoring limits $L < U$ and a number of subgroups $M$. These choices affect the performance of the methods, which is well-acknowledged in the literature (see \cite{evans2020statistically}, \cite{covington2021unbiased}, and \cite{ferrando2022parametric}). These articles focus on bias correction to obtain unbiased point estimates and valid confidence intervals. Our context is different: since we are concerned with hypothesis testing, the bias correction takes the form of correcting critical values to ensure appropriate type I errors. 

Differentially private likelihood ratios $\Lambda^\ast_{10}$ are not distributed as their non-private counterparts. Therefore, the critical values for rejecting $H_0$ with $\Lambda^\ast_{10}$ must be bias-corrected to guarantee adequate type I error rates. 

We can find corrected critical values for $\Lambda^\ast_{10}$ via simulation. In the linear regression model, $\Lambda^\ast_{10}$ depends on the data only through $R^2_i$, which under $H_0$ are distributed as beta with shape parameters $p/2$ and $(b_i-p-p_0)/2$. We can repeatedly simulate $R^2_i$ under $H_0$, compute  $\Lambda^\ast_{10}$, and estimate the correct critical value with empirical quantiles. In contexts other than linear regression models, the adjustments can be done approximately since, under mild conditions, Wilks' theorem establishes that the transformed likelihood ratio statistic $2 \log \Lambda_{10}$ is asymptotically distributed as chi-squared.

In the remainder of this section, we quantify the effects of censoring and splitting the data in terms of bias-variance trade-offs. In Appendix~\ref{suppl:hsb2cens}, we present the results of a simulation study that aims to quantify these trade-offs. Our conclusion is consistent with what we report here.

Given the symmetric nature of logarithms of Bayes factors (explained in Section~\ref{sec:review2}), we specify symmetric censoring limits $L = -C$ and $U = C$ for $C > 0$. Information criteria with penalties as specified in Proposition~\ref{prop:consDPBF} (a prime example being BIC) are consistent under $H_0$ and $H_1$, so it makes sense to specify symmetric censoring limits for their logarithms as well. On the other hand, log-likelihood ratios are non-negative for nested models (which is the case we consider in Section~\ref{sec:DPBF}), so $L = 0$ is a natural lower bound for them. In what follows, we assume that the confidential statistics are censored with the limits we have described in this paragraph.

Let $\tilde{T} = \sum_{i =1}^n T_i^c/M + \eta$ be an $\varepsilon$-differentially private statistic released by the subsample and aggregate technique (similar arguments can be made for $(\varepsilon, \delta)$-differentially private statistics with $\delta > 0$). Then,
$$
 \mathrm{Var}(\tilde{T}) = \left[ \sum_{i=1}^n \mathrm{Var}(T_i^c) + 2 {\Delta^2}/{\varepsilon^2} \right]/M^2.
$$
The variance of $\tilde{T}$ is clearly decreasing in the number of subgroups $M$, and it is increasing in the sensitivity $\Delta = U-L$ because $\mathrm{Var}(T_i^c)$ is increasing in $\Delta$.

The bias of the differentially private $\tilde{T}$ with respect to $\mathcal{T}$, which is the statistic we would obtain without censoring or splitting the data, can be decomposed into the bias induced by censoring plus the bias induced by splitting the data: 
\begin{align*}
\underbrace{\E (\tilde{T}  - \mathcal{T})}_{\text{total bias}} &=  \underbrace{\mathbbm{E} \left\{ \sum_{i=1}^M T_i^c/M - \sum_{i=1}^M T_i/M \right\}}_{\text{censoring bias}} + \underbrace{\mathbbm{E}\left\{\sum_{i=1}^M T_i/M - \mathcal{T} \right\}}_{\text{data-splitting bias}}.
\end{align*}

The censoring bias is decreasing in the sensitivity $\Delta$ and the number of subgroups $M$. In the case of logarithms of Bayes factors and information criteria, where the censoring is symmetric about zero, the censoring bias is 
$$
\frac{1}{M}\sum_{i=1}^M P(T_i \le -C) \left[-C - \E(T_i \mid T_i \le -C) \right] + P(T_i \ge C) \left[C - \E(T_i \mid T_i \ge C)\right].
$$
The censoring bias would be zero if the $T_i$ were symmetric about zero. However, the distributions of logarithms of Bayes factors and information criteria tend to be asymmetric because they are consistent: under $H_0$, logarithms of Bayes factors and information criteria tend to be negative, favoring the null; under $H_1$, they tend to be positive, favoring the alternative. If the distribution of the $T_i$ is not symmetric about zero, the censoring bias decreases in $\Delta = 2C$ and $M$. In the case of log-likelihood ratio statistics, the censoring bias is
$$
\frac{1}{M} \sum_{i=1}^M P(T_i \ge U)[U - \E(T_i \mid T_i \ge U)],
$$
which is decreasing in the upper censoring limit $U$ and the number of subgroups $M$. 

The data-splitting bias is harder to characterize generally but, through examples and simulations, we observe that it is typically increasing in the number of subsets. In the case that the subgroups are balanced ($b_i = b$ for all $i \in \{1, 2, \, ... \, , M\}$) the data-splitting bias is
$$
\mathbbm{E}\left\{\sum_{i=1}^M T_i/M - \mathcal{T} \right\}  =  \E(T_i - \mathcal{T}),
$$
where $T_i$ is the same statistic as $\mathcal{T}$ but computed with fewer observations. If $T_i$ and $\mathcal{T}$ were functions like the sample mean, this term would be zero, but for our statistics, this bias will be, in general, non-zero. Intuitively, if the number of subgroups is large and $T_i$ is based on a smaller sample, the bias will tend to increase. This phenomenon occurs empirically in Section~\ref{subsec:hsb2} and Appendix~\ref{suppl:hsb2cens}, where the bias increases as the number of subgroups increase. For a more analytical argument, the example below considers the case where $T_i$ and $\mathcal{T}$ are log-likelihood ratios. In this case, we are able to find the data-splitting bias under the null hypothesis. This example is relevant because both information criteria and Bayes factors depend on the data only though log-likelihood ratios.

\begin{example}  Assume the subgroups are balanced ($b_i = b$ for all $i \in \{1, 2, \, ... \, M\}$), so we have
$$
T_i = \log \Lambda_{10, i} = -b \log(1- R^2_i)/2, \qquad  \mathcal{T} = \log \Lambda_{10} = -n \log(1- R^2)/2. 
$$
Under $H_0$, $1-R^2_i \sim \mathrm{Beta}((b-p-p_0)/2, p/2)$ and $1-R^2 \sim \mathrm{Beta}((n-p-p_0)/2, p/2)$. Using properties of the beta distribution, the bias under $H_0$ is 
$$
\E(T_i - \mathcal{T}) = \frac{b}{2} \left[ \psi\left( \frac{b-p_0}{2} \right) - \psi\left( \frac{b-p-p_0}{2} \right) \right] - \frac{n}{2} \left[ \psi\left( \frac{n-p_0}{2} \right) - \psi\left( \frac{n-p-p_0}{2} \right) \right],
$$
where $\psi$ is the digamma function \citep{abramowitz1964handbook}. Using properties of the digamma function, it is straightforward to see that the bias is positive and increasing as the number of subgroups increases (i.e., as $b$ decreases). The bias is zero, as it should be, when there is no data-splitting and $b = n$.
\end{example}

\noindent In conclusion, we observe two bias-variance tradeoffs:
\begin{enumerate}
    \item \textbf{Censoring:} Increasing the sensitivity $\Delta = U - L$ decreases the bias but increases the variance of the differentially private statistic.
    \item \textbf{Data-splitting:} Increasing the number of subgroups $M$ decreases the variance of the differentially private statistic. The relationship between $M$ and the total bias is slightly more complicated: increasing the number of subgroups decreases the censoring bias, but it increases the data-splitting bias. In our empirical evaluations, we systematically observe that the total bias of the statistic is increasing in $M$.
\end{enumerate}





\color{black}

\subsection{Application: High School and Beyond Survey}
\label{subsec:hsb2}

We analyze a random sample of 200 students from the High School and Beyond survey, which was conducted by the National Center of Education Statistics. We obtained the data from \cite{diez2012openintro}. In \texttt{R}, they are available as \texttt{data(hsb2)} in \texttt{library(openintro)}. 

We consider two hypothesis tests. Both have \texttt{math} scores as their outcome variable $Y$. In the first one, we test if the variable \texttt{gender} (which, in this data set, can take on the values \texttt{male} or \texttt{female}) is predictive of \texttt{math} scores. Under the null hypothesis ($H_{01}$), the model contains only an intercept, whereas under the alternative  ($H_{11}$), the model includes an intercept and \texttt{gender} as a predictor.   Then, we test if \texttt{read} scores are predictive of \texttt{math} scores when \texttt{science} scores are already in the model. In the null hypothesis ($H_{02}$), the model includes an intercept and \texttt{science} scores. Under the alternative hypothesis ($H_{12}$),  the model contains an intercept, \texttt{science} scores, and \texttt{read} scores.

\begin{figure}[h!]
    \centering
    \includegraphics[width= \textwidth]{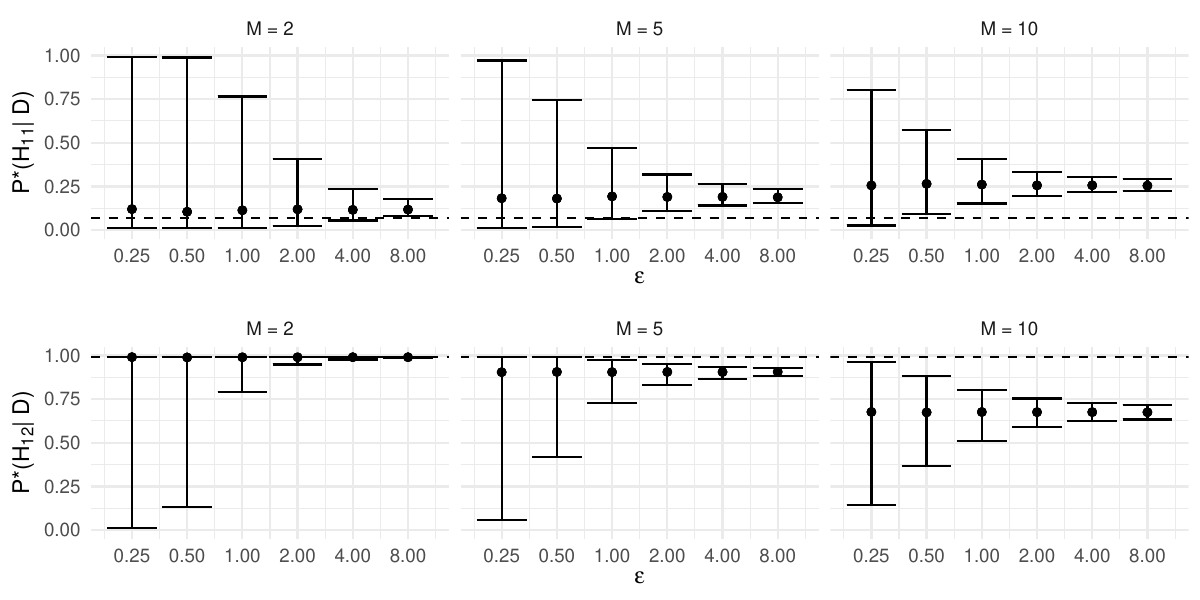}
  \caption{ Distribution of $P^\ast(H_{11} \mid D)$ and $P^\ast(H_{12} \mid D)$  as a function of $\varepsilon$ and $M$. The lower endpoint of the error bars is the first quartile of the distribution, the midpoint is the median, and the upper endpoint is the third quartile. The dashed lines are the true confidential posterior probabilities.}
 \label{fig:postprobs}
\end{figure}

In Figure~\ref{fig:postprobs}, we show the distribution of the differentially private posterior probabilities $P^\ast(H_{11} \mid D)$ and $P^\ast(H_{12} \mid D)$ for different combinations of $\varepsilon$ and $M$, simulating random data splits and perturbation terms $\eta \sim L_1(0, \Delta/(M \varepsilon))$. Our prior probabilities on the hypotheses are $P(H_{01}) = P(H_{02}) = 0.5$, and the prior on the regression coefficients is Zellner's $g$-prior with $g_i$ equal to the sample sizes of the subgroups $b_i$. We censor logarithms of Bayes factors at $L = \log((1-0.99)/0.99)$ and $U = \log((1-0.01)/0.01)$, which is equivalent to censoring posterior probabilities of hypotheses at 0.01 and 0.99, respectively.  Each error bar has been computed with $10^4$ simulations. 

As expected, increasing $\varepsilon$ and $M$ reduces the variability in $P^\ast(H_{11} \mid D)$ and $P^\ast(H_{12} \mid D)$. Additionally, $M$ induces a conservative bias, in the sense that the posterior probabilities are shrunk towards 0.5. For example, if $M = 10$ and $\varepsilon$ is large, $P^\ast(H_{11} \mid D) \approx 0.25$ and $P^\ast(H_{12} \mid D) \approx 0.70$, which are more conservative than the confidential answers $P(H_{11} \mid D) \approx 0.07$ and $P(H_{12} \mid D) \approx 0.99$.

\begin{figure}[h!]
    \centering
    \includegraphics[width= \textwidth]{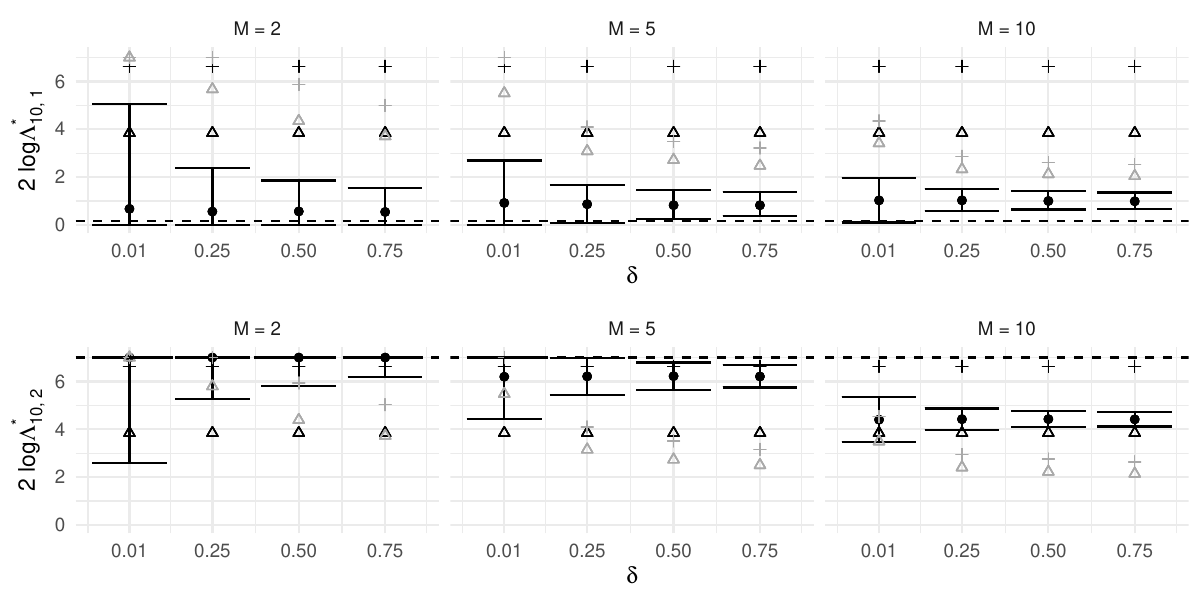} 
    \caption{Distribution of  $2 \log \Lambda^{\ast }_{10, 1}$  and $2 \log \Lambda^{\ast }_{10, 2}$ as a function of $\delta$ and $M$. The gray $+$ and $\vartriangle$ are corrected critical values at the 0.01 and 0.05 significance levels, respectively. The black $+$ and $\vartriangle$ are uncorrected critical values. Dashed horizontal lines are $2 \log \Lambda_{10, 1}$ (first row) and $U = 7$ (second row).}
    \label{fig:calibLR}
\end{figure}  

In Figure~\ref{fig:calibLR}, we show the distributions of $2 \log \Lambda^{\ast}_{10, 1}$ and  $2 \log \Lambda^{\ast}_{10, 2}$ for $L = 0$, $U = 7$, $\varepsilon = 1$, and different values of $0 < \delta < 1$. We treat the data as fixed and quantify the uncertainty introduced by splitting the data randomly and the perturbation term $\eta$, which in this case is normally distributed. The non-private test of $H_{01}$ against $H_{11}$ does not reject $H_{01}$ at the $0.01$ or $0.05$ significance levels, whereas the non-private test of $H_{02}$ against $H_{12}$ rejects $H_{02}$ at the $0.01$ and $0.05$ significance levels. In fact, $2 \log \Lambda_{10, 2}$ is greater than $U$.

As we argued in Section \ref{subsec:cens}, the critical values for rejecting the null hypothesis have to be corrected to take the censoring and data-splitting into account. When $\delta$ is small and $M \in \{2, 5\}$, the corrected critical values (gray $+$ and $\vartriangle$) are larger than the uncorrected ones (black $+$ and $\vartriangle$), so if we did not correct the critical values, we would have inflated type I errors. For larger $\delta$ and $M = 10$, the corrected critical points are lower than the uncorrected ones, showing that correcting critical values can increase power. The distribution of $2 \log \Lambda^{\ast}_{10, 1}$ is generally below the corrected critical values, which coincides with the confidential decision of not rejecting the null. In the case of $2 \log \Lambda^{\ast}_{10, 2}$, corrected tests would reject the null most of the time, in agreement with the confidential answer, especially when $M > 2$ and $\delta \ge 0.25$.

\textcolor{black}{{In Section~\ref{suppl:hsb2cens} of the Appendix, we repeat the simulation study with different censoring limits. Decreasing the sensitivity $U-L$ decreases the variance, but it can introduce bias and lead to tests that are not powerful. In the next subsection, we give simple guidelines to set the censoring limits.}}

\subsection{Guidelines} 

In Section~\ref{subsec:hsb2}, we observed a bias-variance tradeoff when choosing the number of subgroups $M$: as $M$ increases, the variability introduced by $\eta$ decreases at the expense of introducing some bias. 

Before obtaining a private statistic, users can simulate mean squared errors to decide the value of $M$ that is most appropriate in their application. The simulations can be run for different values of $R^2$. Another option is using the sampling distribution of $R^2$, which is known, and average our uncertainty over it. 

\textcolor{black}{The bounds $L$ and $U$ can be values of the confidential statistic that users consider small or large enough.  Stringent censoring decreases the variance of the output, but it can introduce substantial bias if the true posterior probability or likelihood ratio lies outside the censoring limits. In Section~\ref{subsec:hsb2}, we chose $L$ and $U$ that were equivalent to censoring posterior probabilities at $0.01$ and $0.99$, respectively, since those values indicate strong support against or in favor of hypotheses. Practictioners can also use Jeffreys' scale of evidence (Table~\ref{tab:jeffreys}) to censor Bayes factors instead of posterior probabilities. In the case of likelihood ratios, we recommend setting the lower censoring bound $L$ to 0, which is natural in this case, and let the upper censoring value be greater than the critical value for rejection of the non-private test to avoid loss of power. Setting $U$ to approximately twice the critical value performed well in our applications.}

\section{Model Averaging and Selection} \label{sec:modunc}

In this section, we work with the normal linear model as defined in Section~\ref{sec:DPBF}, but instead of comparing two nested models, we consider all the $2^p$ models that can be constructed by taking subsets of columns of $X$. We express our model uncertainty through a binary vector $\gamma  = (\gamma_1, \, \gamma_2, \, ... \, , \gamma_p)' \in \{0,1\}^p$ such that $\gamma_j = 0$ if and only if $\beta_j = 0$. We use the notation $|\gamma| = \sum_{j=1}^p \gamma_j$ for the number of active coefficients in a model. 

Given a model indexed by $\gamma$, we write the linear model as 
$$
Y = X_0 \beta_0 + V_\gamma \beta_\gamma + \sigma W, \, \, \, W \sim N_n(0_n, I_n),
$$
where $\beta_{\gamma} \in \mathbbm{R}^{|\gamma| \times 1}$ is a vector including the $\beta_j$ such that $\gamma_j = 1$ and $V_{\gamma} \in \mathbbm{R}^{n \times |\gamma|}$ is a matrix with the active variables in $\gamma$. Just as we did in Section~\ref{sec:DPBF}, we parameterize the model so that $X_0$ and $V_{\gamma}$ are orthogonal. For convenience, we denote the null model where none of the variables are active as $\gamma = 0$. The matrix $X_0$ represents a set of predictors we are sure to include in our model. Usually, $X_0$ contains an intercept $1_n$, but it can be empty as well. If $X_0$ is empty, the methods in this section would be defined in analogous manner after replacing $I_n - P_{X_0}$ by $I_n$.

From a Bayesian perspective, our prior specification on the regression coefficients $\beta_\gamma$ is the same we had in Section~\ref{sec:DPBF}: we put mixtures of $g$-priors on $\beta_\gamma \mid  \sigma^2, \gamma$ and the right-Haar prior $\pi(\psi, \sigma^2) \propto 1/\sigma^2$ on the common parameters. If $\gamma$ is not the null model, the expression for the non-private Bayes factor of model $\gamma$ to the null model, denoted $B_{\gamma 0}$, is identical to the one in Equation~\eqref{eq:bf} after substituting $p$ by $|\gamma|$ and $R^2$ by $R^2_{\gamma} = Y'P_{V_{\gamma}} Y/Y'(I-P_{X_0})Y$. If $\gamma$ is the null model, we have $B_{\gamma 0} = 1$.

Given a prior  distribution $\pi(\gamma)$ on $\gamma$, the posterior probability of the model identified by $\gamma$ given the data $D$ is defined as 
$$
P(\gamma \mid D) = P(\gamma) \, B_{\gamma 0}/ \sum_{\tilde{\gamma} \in \{0,1\}^p} P(\tilde{\gamma}) \, B_{\tilde{\gamma} 0},
$$
which depends on the data only through the Bayes factors $B_{\gamma 0}$. We assume that the prior $P(\gamma)$ can depend on $p$, but not on $n$ or the design matrix. Most common choices for $P(\gamma)$, like a uniform prior $P(\gamma) = 2^{-p}$ or the hierarchical uniform prior recommended by \cite{scott2010bayes} satisfy the condition. 

From a non-Bayesian perspective, the information criteria  $I^\ast_{\gamma 0}$ are as in Equation~\ref{eq:info} after substituting $R^2$ by $R^2_{\gamma}$ and $\rho$ by $\rho_{|\gamma|}$, where $\rho_{|\gamma|}$ is an increasing function in $|\gamma|$ such as $\rho_{|\gamma|} = |\gamma|$. 

We could use the method in Section~\ref{subsec:lm} to release differentially private versions of all the $B_{\gamma 0}$ or $I_{\gamma 0}$. However, if we released those $2^p$ statistics, the variance of the private statistics would increase exponentially in $p$. Instead, we propose working with a perturbed version of a sufficient statistic whose dimension increases quadratically in $p$. 

Let $Z = (I - P_{X_0})Y \in \mathbbm{R}^{n \times 1}$ be the centered outcome variable and $V \in \mathbbm{R}^{n \times p}$ be the design matrix of the full model that includes all $p$ predictors, \textcolor{black}{which we collect in a data matrix $D = [V; Z]$}. Assuming $Z'Z > 0$ almost surely, we define 
$$
G = D'D = \begin{bmatrix} V'V & V'Z \\ 
Z'V & Z'Z
\end{bmatrix} = \begin{bmatrix} V'V & V'Y \\ 
Y'V & Y'(I-P_{X_0})Y
\end{bmatrix}.
$$
\color{black}{The Gram matrix $G$ is a sufficient statistic for the normal linear model (see, for example, \cite{seber2012linear}).} As a consequence, all of the $R^2_\gamma$, $B_{\gamma 0}$, and $I_{\gamma 0}$ can be constructed by taking appropriate subsets of $G$.

We propose releasing a noisy version of the sufficient statistic $G$, a technique known in the differential privacy literature as \textit{sufficient-statistic perturbation} (see, for example, \cite{mcsherry2009differentially, vu2009differential} and \cite{bernstein2019differentially}).

\color{black}

We construct a differentially private version of $G$ by adding a random perturbation term, defining $G^\ast = G + E$,
where $E$ is a random  perturbation matrix that ensures differential privacy. To ensure $\varepsilon$-differential private, we use the Laplace mechanism. For $(\varepsilon,\delta)$-differential privacy with $0 < \varepsilon < 1$ and $0 < \delta \le 1$, we use Algorithm 2 in \cite{sheffet2019old}, which we refer to as the Wishart mechanism. 

To establish the parameters of the distribution of $E$, we assume that there are lower and upper bounds for the data: that is, there are $l$ and $u$ such that each entry $d_{ij}$ in $D$ is within the interval $[l, u]$. Since the response and predictors are centered, $l<0$ and $u>0$. The entries of $G = D'D$ are of the form $\sum_{i=1}^n d_{ij}d_{ij'}$. Here, $d_{1j}$ is the first row and $j$th column of $D$. If we replace this entry by another one, say $\tilde d_{1j'}$, then the maximum absolute difference in the entries of $G$ is $|d_{1j}d_{1j'} - \tilde d_{1j} \tilde d_{1j'}|$. When computing the sensitivity $\Delta_1$, we only consider the entries where $j \ge j'$ because $G$ is symmetric. Therefore, the sensitivity $\Delta_1$ can be upper-bounded as follows:
\begin{eqnarray*} 
\Delta_1 & = &
\sup_{D \sim \widetilde{D}} 
\Vert D'D -\tilde D' \tilde D \Vert_1
 \\
 & = & 
\sup_{D \sim \widetilde{D}} 
\sum_{j \geq j'} |d_{1j}d_{1j'} - \tilde d_{1j} \tilde d_{1j'}| 
  \\
 & = & 
\sum_{j \geq j'} \sup_{D \sim \widetilde{D}} 
|d_{1j}d_{1j'} - \tilde d_{1j} \tilde d_{1j'}| 
  \\
& \leq & 
(p+1)(p+2) \sup_{d_{11},\,d_{11'} \in (u,l)}|d_{11}d_{11'}|\\
& \leq & (p+1)(p+2) (l^2 \vee u^2).
\end{eqnarray*}
For the Laplace mechanism, we define the perturbation term $E$ as a symmetric random matrix of the form
$$
E = 
\begin{bmatrix}
e_{11} & e_{12} &... & e_{1, p+1} \\
e_{12} & e_{22} & ... & e_{2, p+1} \\
\vdots & \vdots & \ddots & \vdots \\
e_{1, p+1} & e_{2, p+1} & ... & e_{p+1, p+1}
\end{bmatrix},
$$ 
where $e_{jj'} \stackrel{\rm iid}{\sim} L_1\left(0,  \Delta_1/\varepsilon \right)$, for  $j  \ge j'$, and $e_{j'j} = e_{jj'}$.
For the Wishart mechanism \cite[i.e., Algorithm 2 in][]{sheffet2019old}, 
we need a uniform bound on the Euclidean norm of the rows of $D$. Given the assumption $l<d_{ij}<u$, we can use $\sqrt{(p+1)(l^2 \vee u^2)}$ as a uniform bound. In this case, given the bound, we define the random perturbation as $E = M - \mathbb{E}(M)$, where $M \sim W_{p+1}(k, (p+1)(l^2 \vee u^2) I_{p+1})$ and $k = \lfloor p+1 + 28 \log(4/\delta)/\varepsilon^2 \rfloor$ with $0 < \varepsilon < 1$ and $0 < \delta \le 1$.





\color{black}

For both mechanisms, the variances of the entries in $E$ can be high, especially for small values of $\varepsilon$ or $\delta$. This, in turn, can lead to outputs that overestimate the number of active predictors. To avoid this issue, we propose post-processing $G^\ast$ in two ways: hard-thresholding off-diagonal elements as in \cite{bickel2008covariance} and adding a constant $r$ to the diagonal elements. 

We propose thresholding the off-diagonal entries of $G^\ast$ at $e_{\lambda}$, the $\lambda$-th percentile of $E_{ij}$. More formally, we define \textcolor{black}{$G^{**}$ to be the matrix with typical element $G^{\ast}_{ij} \mathbbm{1}( i = j \text{ or } |G^{\ast}_{ij}| \ge e_{\lambda})$}. This choice can be justified as follows: if an off-diagonal entry of $G^{\ast}_{ij} = G_{ij} + E_{ij}$ is not an extreme value in the distribution of $E_{ij}$, it is likely that $G^\ast_{ij}$ is essentially $E_{ij}$ and $G_{ij}$ is nearly zero.

Adding a constant $r$ to the diagonal elements of $G^\ast$ can be seen as a ridge-type of regularization that can reduce the variability of the output. It can also be used to guarantee that the output is positive-definite because, after adding the Laplace perturbation terms, $G^\ast$ may not be positive-definite. In fact, the hard-thresholded matrix $G^{\ast \ast}$ need not be positive-definite even if $G^\ast$ is positive-definite \citep{bickel2008covariance}. Given a symmetric matrix $A$, which here can be either $G^\ast$ or $G^{\ast \ast}$, the matrix $A_r = A + r I_{p+1}$ is positive-definite as long as $r > - \mathrm{eig}_{\mathrm{min}}(A)$, where $\mathrm{eig}_{\mathrm{min}}( \cdot )$ is a function returning the minimum eigenvalue of a matrix. 

In our applications (in Section~\ref{subsec:empeval}), we consider methods that hard-threshold $G^\ast$ and methods that do not. That is, we compare methods that are based on $G^{\ast \ast}_r = G^{\ast \ast} + r I_{p+1}$ to methods based on $G^\ast_r = G^\ast + r I_{p+1}$, where $G^\ast$ is not hard-thresholded. In our experience, hard-thresholding is helpful when the ground truth is sparse, but it can be detrimental when most predictors are active. We justify this argument in Section~\ref{subsec:empeval}.

Proposition~\ref{prop:cons_modunc} establishes model-selection consistency under some assumptions. {\textcolor{black}{More precisely, we show that the differentially private Bayes factor of any model $\gamma$ to the true model, which can be expressed as $B^\ast_{\gamma T} = B^{\ast}_{\gamma0}/B^{\ast}_{T 0}$, converges to zero in probability for any $\gamma \neq T$. The differentially private information criteria $I^\ast_{\gamma T}$ are also consistent under the assumptions listed below.}} We proved the result by characterizing the asymptotic behavior of $G^\ast_r$, $G^{\ast \ast}_r$, and $R^{2, \ast}_{\gamma}$, and then bounding the Bayes factors and information criteria above and below. Just as we had in Section~\ref{sec:DPBF}, the proof covers  $\varepsilon$ and $(\varepsilon, \delta)$ differentially private methods. The proof does not follow from \cite{liang2008mixtures} for several reasons, one of them being that we are not assuming that the response given the covariates is normal, since we assume that the data are bounded.

\begin{prop}\label{prop:cons_modunc} Let $T \in \{0, 1\}^p$ be a vector indexing the truly active predictors. As $n \rightarrow \infty$, and under the regularity conditions listed below, $B^{\ast}_{\gamma T} \rightarrow_P 0$ and $I^{\ast}_{\gamma T} \rightarrow_P 0$ for any $\gamma \neq T$.
\begin{enumerate}
     \item \textbf{Boundedness:}  The data $D$ are within the interval $[l, u]$ for finite $l$ and $u$.
    \item \textbf{Regression mean and variance:} $\mathbb{E}(Y \mid X_0, V) = X_0 \psi + V_T \beta_T$ and $\mathrm{Var}(Y \mid X_0, V) = \sigma^2_T I_n$, where $V_T \in \mathbbm{R}^{n \times p_T}$ is a matrix that contains the truly active predictors.
     \item \textbf{Privacy parameters:}  The privacy parameters $\varepsilon$ and $\delta$ are such that the matrix perturbation term $E/n \rightarrow_P 0.$
    \item \textbf{Regularization parameters:} $\lambda$ is fixed and $r$ is so that $\lim_{n \rightarrow \infty} r/n = 0$.
    \item \textbf{Design matrices:} $\lim_{n \rightarrow \infty} V'V/n = S_1$, where $S_1$ is symmetric and positive-definite.
    \item \textbf{Priors on $g$ and penalties $\rho$:} for all $1 \le |\gamma| \le p$, the prior $\pi(g)$ satisfies
\begin{align*}
 \lim_{n \rightarrow \infty} \int_{0}^\infty  n^{p/2} \, (g+1)^{-|\gamma|/2} \pi(g) \, \nu(\mathrm{d} g) &< \infty \\ \lim_{n \rightarrow\infty}   \int_{n}^\infty  n^{p/2}  (g+1)^{-|\gamma|/2}  \pi(g) \, \nu( \mathrm{d} g) &> 0,
\end{align*}
or from a non-Bayesian perspective, $\rho_{|\gamma|}$ is increasing in $|\gamma|$ and satisfies $|\gamma| \le \rho_{|\gamma|} \le  |\gamma| \vee \log n$. 
\end{enumerate}
\end{prop}

The framework here is different to the one in Section~\ref{sec:DPBF}. We assume that the  confidential data are bounded (Assumption 1) and do not assume that the distribution of the outcome given the predictors is normal. In other words, Bayes factors and information criteria are misspecified beyond the addition of the perturbation term $E$. Nonetheless,  Proposition~\ref{prop:cons_modunc} shows that the methods are consistent. Our setup is also distinct to the one adopted in \cite{lei2018differentially}, where it is simultaneously assumed that the response is normal (Assumption 1 in \cite{lei2018differentially}) and bounded (Assumption 4).  Assumption 2 requires  that the regression mean be well-specified and that the covariance of the response given the predictors be spherical. Assumption 3 forces the perturbation matrix $E$ to be so that $E/n \rightarrow_P 0$. As we had in Section~\ref{sec:DPBF}, it suffices to let $\varepsilon$ and $\delta$ be fixed for it to hold. Assumption 4 imposes conditions on the regularization parameters. The case of a non-thresholded matrix $G^\ast_r$ is included as $\lambda = 0$. Assumption 5 is similar to the regularity condition on design matrices in Proposition~\ref{prop:consDPBF}. Finally, Assumption 6 is essentially the same as the assumptions on the priors and penalties in~\ref{prop:cons_modunc}. Just as we had in Section~\ref{sec:DPBF}, the differentially private Bayes factors with Zellner's $g$-prior with $g = n$, the robust prior, Zellner-Siow, and BIC are all consistent, and so is BIC.

{\color{black}{In the common scenario where $X_0$ is an intercept $1_n$, the methods described in this section can be conveniently implemented in \texttt{R} with the \texttt{bas.lm} function in \texttt{library(BAS)}. Given a set of predictors and an outcome variable, the \texttt{bas.lm} function enumerates Bayes factors for small to moderate $p$ and samples from the model space for large $p$. The function outputs other statistics of interest such as posterior inclusion probabilities and model-averaged estimates. To use \texttt{bas.lm} for our problem, we need to generate a synthetic dataset $\mathcal{D}$ (containing both centered predictors and outcome) whose sufficient statistic $\mathcal{D}'\mathcal{D}$ is equal to a fixed Gram matrix $\mathcal{G},$ which can be $G^\ast_r$ or $G^{\ast \ast}_r$. Proposition~\ref{prop:synthetic} below shows how to obtain such a synthetic data set $\mathcal{D} = [\mathcal{V}; \mathcal{Z}]$ given a Gram matrix $\mathcal{G}$.

\begin{prop} \label{prop:synthetic}
Let $\mathcal{U} \in \mathbbm{R}^{n \times (p+1)}$ be a full-rank matrix and define $\mathcal{M} = (I_n-1_n1_n'/n) \mathcal{U}$. Given a Gram matrix $\mathcal{G}$, we can generate a synthetic dataset $\mathcal{D} = [\mathcal{V}; \mathcal{Z}]$ with  $\mathcal{D} = \mathcal{M} (\mathcal{M}'\mathcal{M})^{-1/2} \mathcal{G}^{1/2}$.   The synthetic dataset $\mathcal{D}$ satisfies the identities $\mathcal{D}'\mathcal{D} = \mathcal{G}$, $\mathcal{V}' 1_n = 0_{p}$, and $\mathcal{Z}' 1_n = 0$. 
\end{prop}

Proposition~\ref{prop:synthetic} guarantees that the outputs we obtain from  running \texttt{bas.lm} on the synthetic data $\mathcal{D}$ are identical to what we would find by taking subsets of $\mathcal{G}$ directly. In the proposition, the matrix $\mathcal{U}$ is arbitrary, but in practice its entries can be simulated by sampling independently from the uniform distribution.

\color{black}
\subsection{Quantifying the Uncertainty Introduced by the Mechanism}
\label{subsec:hist}

\color{black}
With the non-thresholded methods, the private statistics are of the form
$G^{\ast}_r = G + E +r I_{p+1}$. Since $r$ and the distribution of $E$ are both known, we can define a confidence set for the non-private $G$ given $G^{\ast}_r$. With such a set, it is possible to find confidence regions for summaries of interest $T(G)$ like least-squares estimates or inclusion probabilities.

To define a $1-\alpha$ confidence set for $G$, we first find $\mathcal{E}_{1-\alpha}$ such that $P(E\in\mathcal{E}_{1-\alpha})=1-\alpha.$ For a matrix norm $\lVert \cdot \rVert$, define $\mathcal{E}_{1-\alpha} = \{ E : \lVert E - \mathbb{E}(E) \rVert  \le q_{1-\alpha}\}$ where $P[ \lVert E - \mathbb{E}(E) \rVert \le q_{1-\alpha}] = 1-\alpha$. Then, $\mathcal{\tilde{C}}_{1-\alpha} = \{ G^{\ast}_r - r I_{p+1} - E: E \in \mathcal{E}_{1-\alpha} \}$ is a $1-\alpha$ confidence set for $G$. Since $G$ is symmetric and positive-definite, we can intersect $\mathcal{\tilde{C}}_{1-\alpha}$ with the set of symmetric positive-definite matrices $\mathcal{S}_{++}$ to define a $1-\alpha$ confidence region $\mathcal{C}_{1-\alpha} = \mathcal{\tilde{C}}_{1-\alpha} \cap \mathcal{S}_{++}$ whose volume is at most that of $\mathcal{\tilde{C}}_{1-\alpha}$. The confidence set $\mathcal{C}_{1-\alpha}$ can be transformed into $T(\mathcal{C}_{1-\alpha})$ to produce confidence sets for summaries of interest $T(G)$. 


We can approximate $T(\mathcal{C}_{1-\alpha})$ with a rejection sampler. First, simulate $E_{1}, E_{2}, \, ... \, , E_{n_\mathrm{sim}}$ from the appropriate mechanism and compute $\lVert E_i - \mathbb{E}(E) \rVert$ for $i \in \{1,2, \, ... \, , n_{\mathrm{sim}}\}$. Then, approximate the $((1-\alpha) \times 100)$th percentile $q_{1-\alpha}$ with its empirical version $\hat{q}_{1-\alpha} = \inf \{q :  \sum_{i=1}^{n_\mathrm{sim}} \mathbbm{1}(\lVert E_i - \mathbb{E}(E) \rVert \le q)/n_{\mathrm{sim}} \ge 1-\alpha\}$ and define $\hat{\mathcal{E}}_{1-\alpha} = \{ E_i: \lVert E_i - \mathbb{E}(E) \rVert \le \hat{q}_{1-\alpha}, \, i = 1, \ldots, n_\mathrm{sim}\}$. After that, we can find $T(\mathcal{\hat C}_{1-\alpha})$; that is, compute $T(G^{\ast}_r - r I_{p+1} - E_i)$ for $E_i \in \hat{\mathcal{E}}_{1-\alpha}$ and only keep those such that $G^\ast_r - r I_{p+1} - E_i \in \mathcal{S}_{++}$.

In general, the confidence set need not be an interval, but we can summarize the confidence set with a histogram. To do so, we define the bins of the histogram as $\mathcal{B}_k = [t_{k-1},t_k)$
with $\min\{T(\mathcal{\hat C}_{1-\alpha})\} = t_0 < t_1 < \ldots < t_{K-1} < t_K = \max\{T(\mathcal{\hat C}_{1-\alpha})\}$ and their corresponding relative frequencies $\#\mathcal{B}_k$ using the number of elements of $T(\mathcal{\hat C}_{1-\alpha})$ that fall in $\mathcal{B}_k$, $k = 1,\ldots,K.$ We denote the histogram summarizing $T(\mathcal{\hat C}_{1-\alpha})$ as $\text{Hist}(T,\mathcal{\hat C}_{1-\alpha}) 
= \{(\mathcal{B}_k,\#\mathcal{B}_k)\}_{k=1}^K$. We can also report histograms in a density scale; that is, $\text{Hist}(T,\mathcal{\hat C}_{1-\alpha}) 
= \{(\mathcal{B}_k, d_k)\}_{k=1}^K$, where 
$$
d_k = \frac{\#\mathcal{B}_k}{|T(\mathcal{\hat C}_{1-\alpha})| (t_k-t_{k-1})}
$$
and $|T(\mathcal{\hat C}_{1-\alpha})|$ denotes the cardinality of $T(\mathcal{\hat C}_{1-\alpha})$.

\color{black}

\subsection{Empirical Evaluations}
\label{subsec:empeval}

We evaluate the performance of the methods described in this section in a simulation study and an application. The simulation study is similar to the one in \cite{liang2008mixtures}, whereas the real data set is a subset of the March 2000 Current Population Survey that was analyzed in \cite{barrientos2019differentially}.

We implement the methods with the \texttt{R} package \texttt{BAS} \citep{BAS}. In the simulation study, we found Bayes factors with the Zellner-Siow prior (ZS) and information criteria with BIC. The prior distribution on the model space $\pi(\gamma)$ is the hierarchical uniform prior proposed in \cite{scott2010bayes}. From a non-Bayesian perspective, $\pi(\gamma)$ acts as a function that weighs the information criteria. The results with Zellner-Siow and BIC are almost identical. We report the outputs based on the former here and show the results with the latter in Appendix~\ref{suppl:plots}.

We compare the results we obtain by hard-thresholding and not thresholding the Gram matrix $G^\ast$. In all cases, we add a regularization parameter $r$ to the diagonal entries of $G^\ast$. For the Laplace mechanism, we set $r$ to be the 99-th percentile of $\mathrm{eig}_{\rm min}(E)$, which we find via simulation. For the Wishart mechanism, we use the analytical expression in Remark 2 of \cite{sheffet2019old}.

\subsubsection{Simulation Study}
\color{black}

We simulate data from a normal linear model with $p$ predictors, where $p$ is set to 2, 6, or 9. The sample size $n$ (in thousands) varies from $5$ to  $10,000$. The number of active predictors in the true model $|T|$ depends on the value of $p$ and ranges from 0 (null model is true) to $p$ (full model is true). Specifically, if $p=2$, we set $|T|\in\{0,1,2\}$; if $p=6$, we set $|T|\in\{0,3,6\}$; and if $p=9$, we set $|T|\in\{0,4,9\}$.  The predictors are independently drawn from the uniform distribution on $(-2,2)$. Following \cite{hastie2017extended}, we define the signal-to-noise ratio (SNR) as the variance of the regression function (which is random, since we are simulating predictors and $\beta$) divided by $\sigma^2$. In our simulations, we assume that the intercept is zero and $\beta$ is a $p$-dimensional vector equal to $b [1,\ldots,1]'$. We use optimization to find $\sigma^2$ and $b$ such that SNR $= 0.5$ and the response falls within $(-2,2)$ with high probability. For each combination of $|T|$ and $n$, we simulate 1,000 data sets. All the data sets we simulated are such that the response falls in $(-2,2)$. We consider $\varepsilon \in \{0.5,0.9\}$ and, in the case of the Wishart mechanism, we set $\delta = 1/n$.

\begin{figure}[h!]
  \centering
  \includegraphics[scale= 0.55]{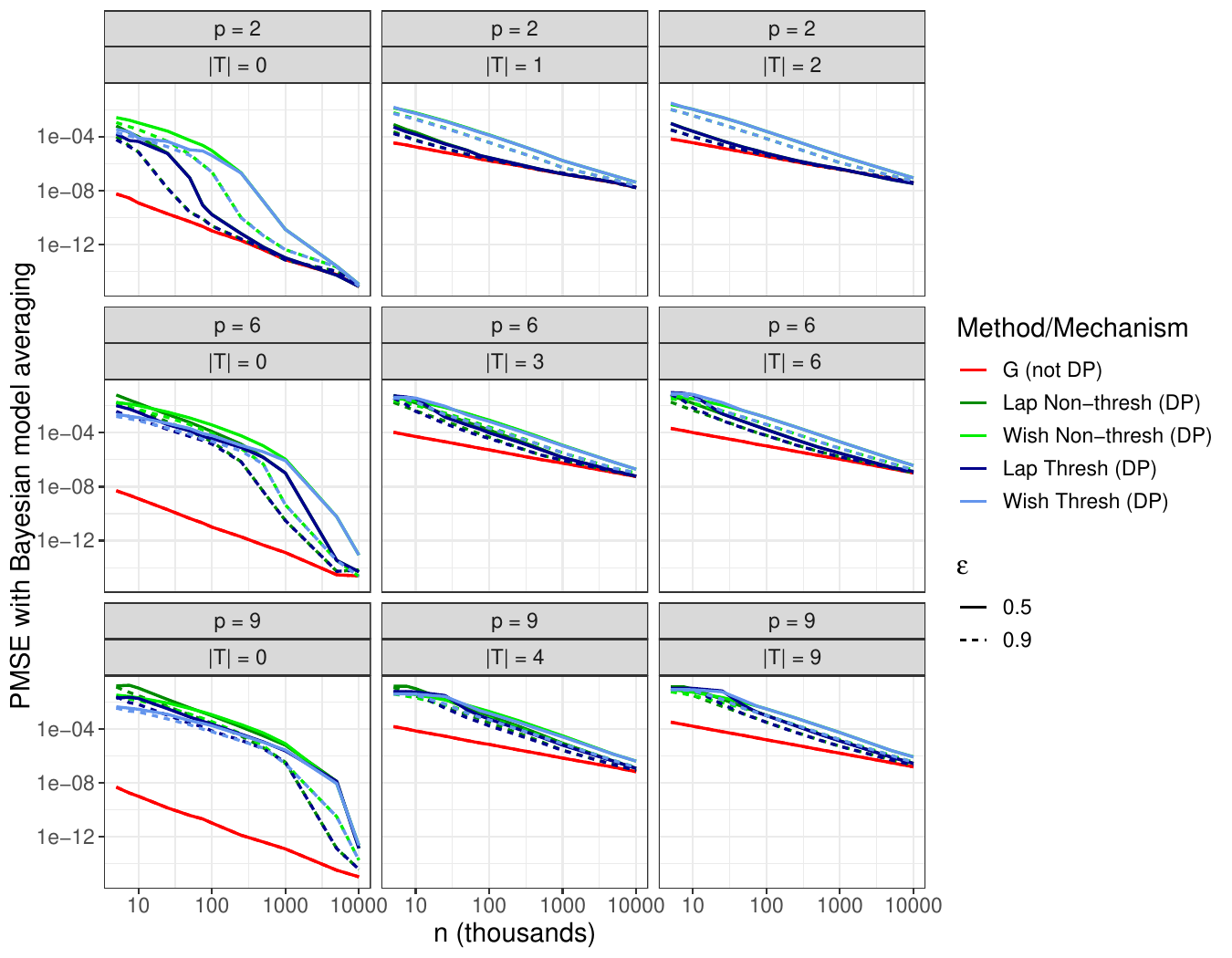}
  \caption{Simulation study: Sample size ($x$-axis) against log(PMSE) ($y$-axis) with Zellner-Siow prior.}
 \label{fig:log_MSE_simulation_ZS}
\end{figure}

We assess the performance of the methods by tracking Monte Carlo averages of predictive mean squared errors and the posterior probability of the true model. We define the predictive mean squared error as $\mbox{PMSE} = n^{-1}\Vert V_T \beta_T - V\beta^\ast \Vert_2^2,$ where $V_T$ is a design matrix containing truly active predictors, $\beta_T$ is the true value of $\beta$, and $\beta^\ast$ is the differentially private model-averaged posterior expectation. 

\begin{figure}[h!]
  \centering
  \includegraphics[scale= 0.55]{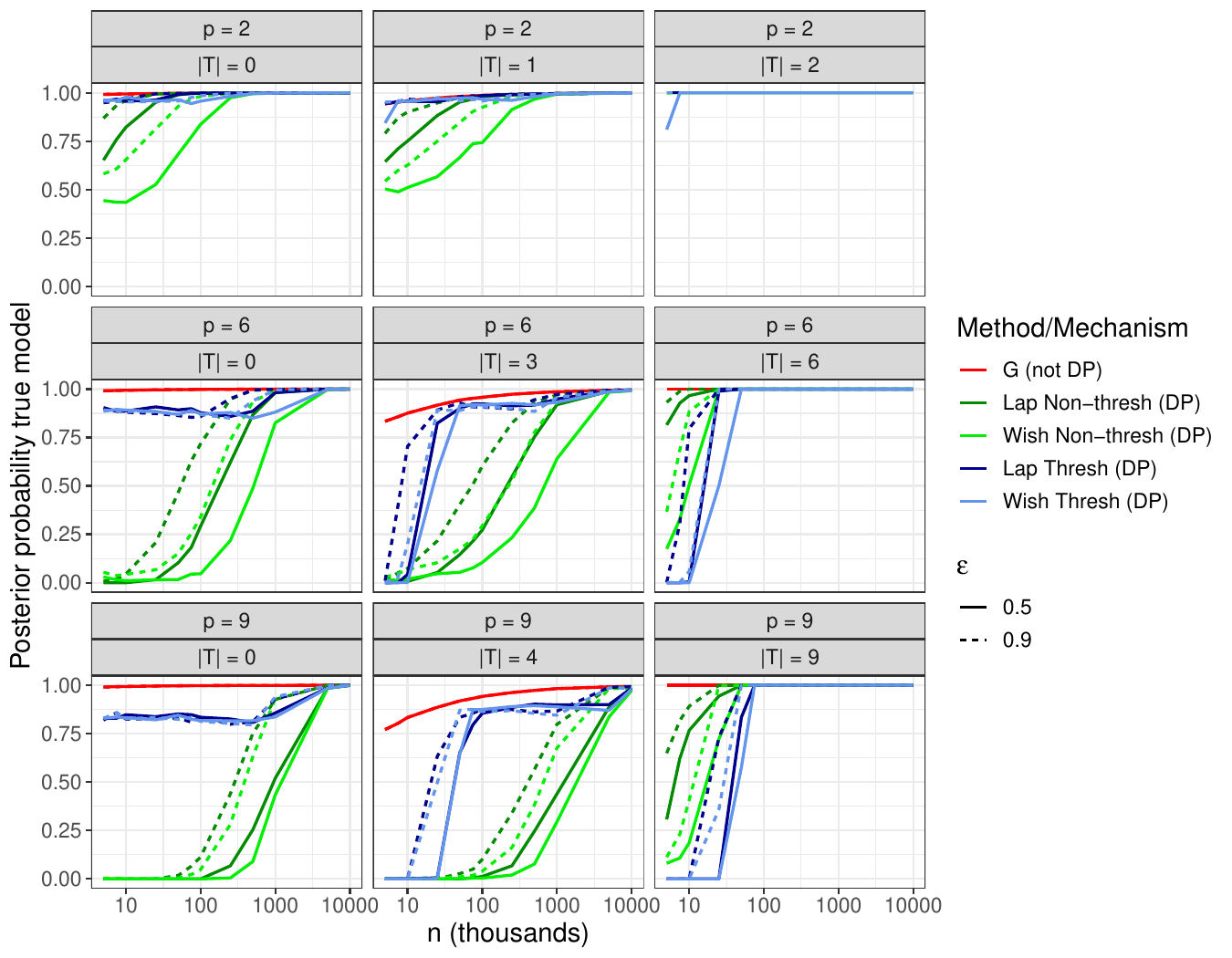}
 \caption{Simulation study: Sample size ($x$-axis) against posterior probability of the true model ($y$-axis) with Zellner-Siow prior. }
\label{suppl:fig:Inclusion_Prob_ZS}
\end{figure}

Figure \ref{fig:log_MSE_simulation_ZS} displays PMSEs for different values of $n$, $p$, $\varepsilon$, and $|T|$. As expected, the PMSEs for both private and non-private approaches decrease as the sample size increases. We observe that the PMSEs of the differentially private methods are smaller when $\varepsilon$ is 0.9 compared to when $\varepsilon$ is 0.5, and they are always higher than the non-private PMSEs. This is expected, as larger values of $\varepsilon$ should lead to greater statistical utility. In most cases, the methods based on the Laplace mechanism have a lower PMSE than those based on the Wishart mechanism. However, the Wishart mechanism seems to perform better in the case when $p$ is either 6 or 9, $|T|$ is zero, and the sample size is small. Although this is slightly less evident in the figure, upon closer inspection, we can see that methods based on hard-thresholding tend to have a slightly lower PMSE when $|T|$ is zero but cease to be advantageous when $|T|$ is large and the sample size is small.

Figure~\ref{suppl:fig:Inclusion_Prob_ZS} displays the posterior probabilities of the true model, for the same values of $n$, $p$, $\varepsilon$, and $|T|$ that we used in Figure~\ref{fig:log_MSE_simulation_ZS}. The results are consistent with what we observed for PMSEs. We also observe that, although all probabilities increase as the sample size increases, the rate at which they increase depends on $p$. For higher values of $p$, the rate of increase is lower because the computational complexity of the problem increases with the dimension of $G$. The number of entries in $G$ increases quadratically in $p$, and so does the variance of the perturbation term added to ensure differential privacy. This fact affects the convergence rate of the posterior probabilities.}

\color{black}

\subsubsection{Application: Current Population Survey}
\label{subsubsec:CPS}

The data set includes $n = 49,436$ heads of households with non-negative incomes. We consider 6 predictors: age in years ($\beta_1$), age squared ($\beta_2$), marital status ($\beta_3$), sex ($\beta_4$), education ($\beta_5$), and race ($\beta_6$). All predictors are numeric or binary except for education, which is an ordinal variable. To reduce the number of coefficients in the model, we treat education as numeric, ranging from 1 (for less than 1st grade) to 16 (for doctoral degree). The binary predictors are: marital status (1: civilian spouse present; 0: otherwise), sex (1: male; 0: female), and race (1: white; 0: otherwise). The response variable is income. 

In this application, the non-private inclusion probabilities are all close to one. To provide a more challenging benchmark for our methods, we permute the rows for marital status and education in the design matrix to artificially make the inclusion probabilities for $\beta_3$ and $\beta_5$ close to zero. \color{black}{The predictors and the response are centered and rescaled to the interval $(-0.5,0.5)$.

\color{black}

\begin{figure}[h!]
    \centering
    \includegraphics[scale=0.375,page=13]
    {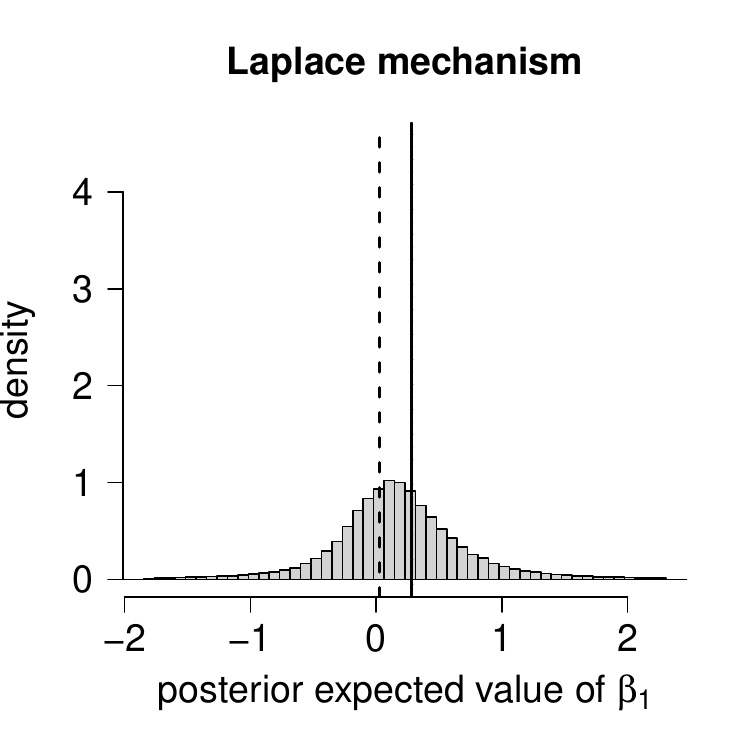}
     \includegraphics[scale=0.375,page=17]
    {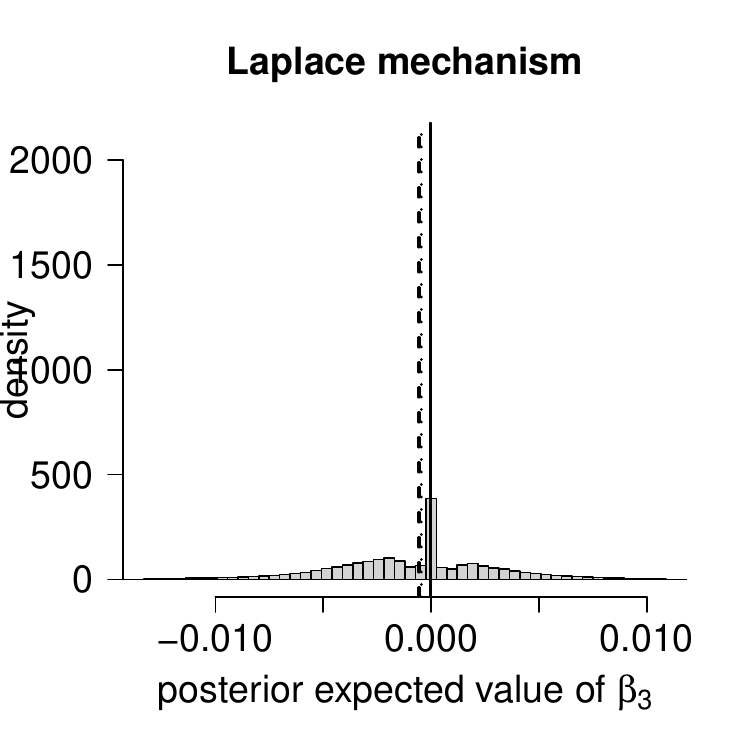}
        \includegraphics[scale=0.375,page=19]
    {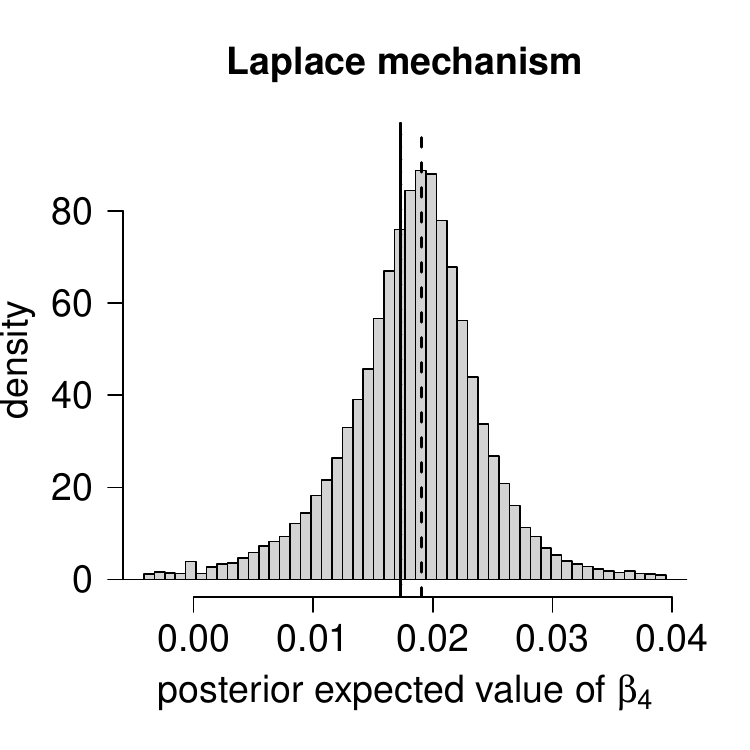}
    \includegraphics[scale=0.375,page=14]
    {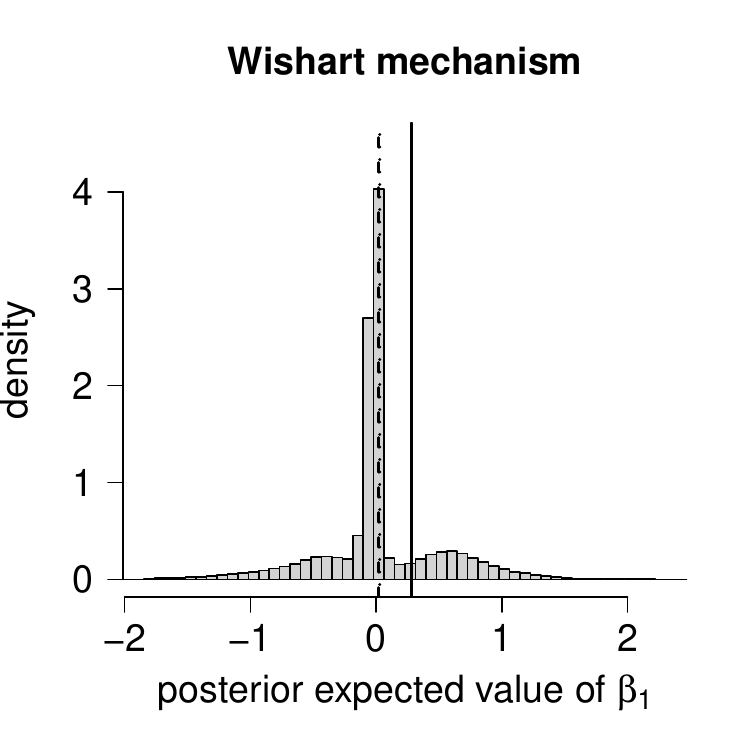}
    \includegraphics[scale=0.375,page=18]
    {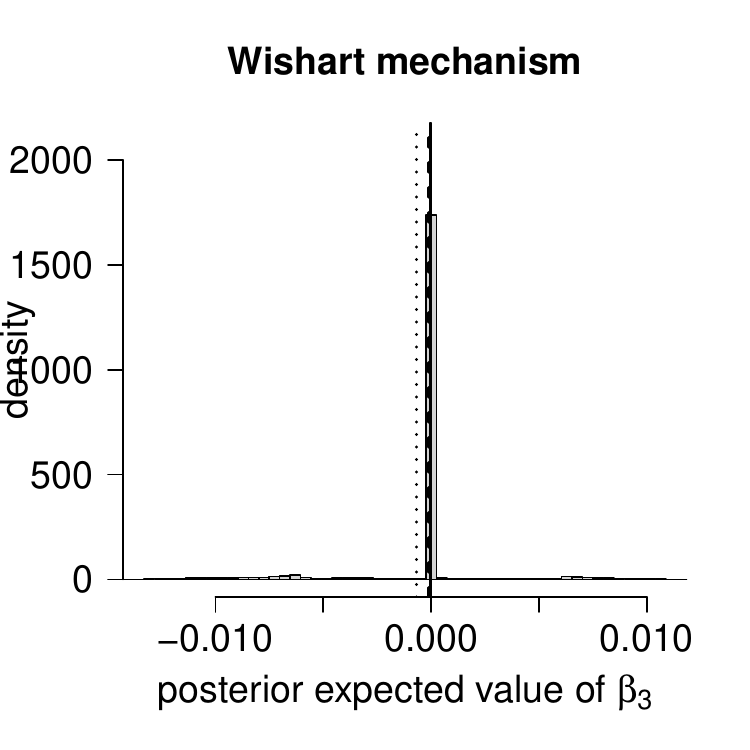}
    \includegraphics[scale=0.375,page=20]
    {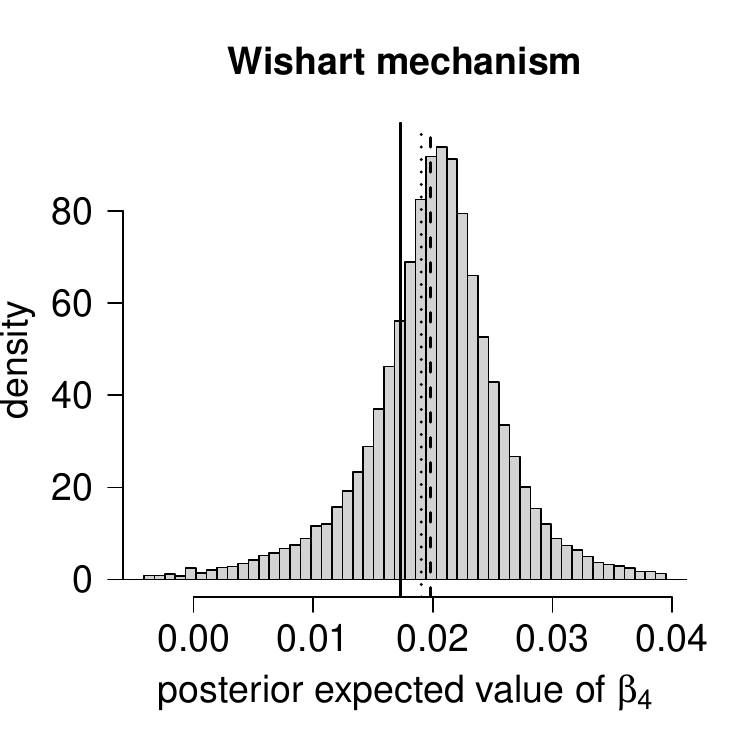}

\caption{Current population survey: Posterior expectations of coefficients $\beta_1, \beta_3$, and $\beta_4$  with Zellner-Siow prior and $\varepsilon = 0.9$.  The vertical solid lines are the non-private posterior expectations, whereas the dashed lines and dotted lines are averaged posterior expectations estimated with non-thresholded methods and thresholded methods, respectively.}
 \label{fig:confidence_hist_ZS_epsilon1_B-1-3-4}
\end{figure}

\color{black}

Figure \ref{fig:confidence_hist_ZS_epsilon1_B-1-3-4} displays the posterior expected values of $\beta_1$, $\beta_3$, and $\beta_4$ with the Zellner-Siow prior and $\varepsilon=0.9$.
\color{black}
We use the histograms described in Section~\ref{subsec:hist} to define approximate 95\% confidence sets for $T(G) = \mathbb{E}(\beta_j \mid G)$. Our choice of matrix norm is the Frobenius norm.  Specifically, we run our procedure 250 times and, for each run and a fixed collection of bins $\mathcal{B}_1,\ldots,\mathcal{B}_K$, we summarize each $T(\mathcal{\hat C}_{1-\alpha})$ with its corresponding histogram  $\text{Hist}(T,\mathcal{\hat C}_{0.95}) 
= \{\left(\mathcal{B}_k,d_k\right)\}_{k=1}^K$.

If we let $d_1^{(l)},\ldots,d_K^{(l)}$  be the densities values of the histogram associated with the $l$-th run, $l=1,\ldots,250$, we define average histograms as 
$$\overline{\text{Hist}}\left(T(\cdot) = \mathbb{E}(\beta_j \mid \cdot), \,\mathcal{\hat C}_{0.95}\right) 
= \left\{\left(\mathcal{B}_k,\overline{d}_k = \frac{1}{250} \sum_{l=1}^{250} d_k^{(l)}  \right)\right\}_{k=1}^K, \quad j \in \{1,3,4\}.
$$
\color{black}
The results are displayed  in Figure \ref{fig:confidence_hist_ZS_epsilon1_B-1-3-4}. 
In all cases, the differentially private methods are close to the non-private answers. We can also see that the histograms are useful for quantifying the uncertainty introduced by the mechanism, since their spread increases when the thresholded and non-thresholded methods do not agree in their estimates.

\subsection{Guidelines}
\label{subsec:guide}

Thresholded methods tend to perform best when the true number of predictors is small. On the other hand, when most predictors are active, non-thresholded methods tend to outperform thresholded methods. The root cause behind this phenomenon is that thresholding shrinks the elements in $G^\ast$, which promotes sparsity.

In practice, we recommend that users run analyses with both thresholded and non-thresholded methods. This can be done without affecting the privacy budget of the analyst because both $G^{\ast}_r$ and the thresholded matrix $G^{\ast\ast}_r$ are post-processed versions of the same differentially private matrix $G^{\ast}$.

Finally, we recommend reporting confidence sets whenever possible, since we find them to be a valuable tool for quantifying the uncertainty introduced by the mechanisms.

\section{Conclusions and Future Work}

\color{black}

In this article, we proposed differentially private methods for hypothesis testing, model averaging, and model selection for linear regression. Under regularity conditions, the methods are consistent. The regularity conditions we have imposed are similar to the conditions used in the literature for establishing consistency of non-differentially private methods. 

Our methods for hypothesis testing are based on data-splitting and censoring statistics. We have studied the effects of these operations on the performance of the methods. In the case of data-splitting, increasing the number of subsets reduces the variance of the differentially private statistics, but it adds bias. In the case of censoring, more stringent censoring reduces the variance, but it can lead to substantial bias if the true, confidential statistic lies outside the uncensored range.

The methods we proposed for model averaging and selection are based on a perturbed sufficient statistic. If we suspect that the ground truth is sparse, we recommend hard-thresholding the perturbed sufficient statistic; however, if most predictors are active, hard-thresholding can lead to underfitting. 

The methodology proposed here could be extended in a number of ways. It would be useful to extend the methods to generalized linear models through the framework proposed in \cite{li2018mixtures}. The implementation for hypothesis testing is straightforward, but for model uncertainty there are no low-dimensional sufficient statistics. This obstacle can be overcome by working with approximate sufficient statistics, as proposed in \cite{huggins2017pass}. This approach has been used successfully in differentially private estimation problems in \cite{kulkarni2021differentially}. 

It would also be interesting to extend the methods to survival models, as health records are confidential. In this case, the extension could adapt the framework proposed in \cite{castellanos2021model}. There are no low-dimensional sufficient statistics for these survival models, either, but one could define approximate sufficient statistics as well. 

In our work, we have used off-the-shelf techniques for establishing differential privacy. While we observe that our proposals can be useful in practice, it might be possible to design more efficient mechanisms that are specifically tailored to the tasks we considered. This is another interesting avenue for future research.

\section*{Acknowledgements}

The authors would like to thank the feedback from two anonymous reviewers that greatly improved the presentation and contents of the article. The research of the second author was supported by the National Science Foundation National Center for Science and Engineering Statistics [49100420C0002 and 49100422C0008] and the Test Resource Management Center (TRMC) within the Office of the Secretary of Defense (OSD), contract \#FA807518D0002.

\color{black}

\clearpage

\bibliography{diffprivacy}
\bibliographystyle{chicago} 

\clearpage

\begin{appendix}

\begin{center}
    \Large \textbf{Appendix}
\end{center}
\normalsize

In Section \ref{suppl:section:Proofs}, we include the proofs of the propositions in the main text, as well as auxiliary results that are helpful for proving them. In Section~\ref{suppl:hsb2cens} we study the effects of censoring and data-splitting in the High School and Beyond Survey application we considered in Section~\ref{sec:DPBF}. Finally, in Section \ref{suppl:plots} we include additional plots for the simulation study and the application in Section 4 of the main text.

\section{Proofs} \label{suppl:section:Proofs}

We use the notation $a \lesssim b$ and $a \gtrsim b$ to denote $\exists K < \infty : a \le K b$ and $\exists K' < \infty : a \ge K' B$, respectively, with the understanding that, if we are taking limits that depend on $n$,  $K$ and $K'$ do not depend on $n$. Unless stated otherwise, vector norms are Euclidean norms; that is, $\lVert x \lVert^2 = \sum_{i = 1}^p x_i^2$ for $x \in \mathbb{R}^p$. We use the $O_p(1)$ notation in our proofs, which we define below.

\begin{definition} \label{def:Op1} Let $X_n$ be a random variable taking values in $\mathbb{R}^n$. Then $X_n \in O_p(1)$ if, for all $\epsilon > 0$, there exist $n_\epsilon, M_\epsilon < \infty$ so that for all $n \ge n_\epsilon$,
$
P [ \, \lVert X_n \rVert^2 > M^2_\epsilon]  < \epsilon.
$
\end{definition}

\subsection{Auxiliary results for Proposition 3}


 

\begin{proposition} \label{prop:tailbounds} [\cite{birge2001alternative}, \cite{laurent2000adaptive}]
Let $X \sim \chi^2_m(\lambda)$, $Y \sim \chi^2_n$ and $x \ge 0$. Then, 
\begin{align*}
    P(X \le m + \lambda - 2 \sqrt{(m+ 2 \lambda) x}) &\le \exp(-x) \\
    P(Y  \ge n + 2 \sqrt{n x} + 2x ) &\le \exp(-x)
\end{align*}
\end{proposition}

\begin{proposition} \label{prop:tailR2} Let $R^2_i \sim \mathrm{Beta}(p/2, (b_i-p-p_0)/2
)$ for $b_i > p_0 + p + 2$. Then, if $p \ge 2$:
$$
 P(R^2_i > k_i) \le \frac{ 2 (1-k_i)^{(b_i-p-p_0)/2}}{(b_i-p-p_0)  \,  \text{B}(p/2,(b_i-p-p_0)/2) }.
$$ 
If $p = 1$:
$$
 P(R^2_i > k_i) \le  \frac{1}{\text{B}(1/2,(b_i-1-p_0)/2)}  \left(\frac{(1-k_i)^{b_i-p_0-2} \log(1/k_i)}{b_i-p_0 -2}\right)^{1/2},
$$
where $\text{B}( \cdot, \cdot)$ is the Beta function.
\begin{proof}
This result is straightforward to prove, but we could not find it in the literature. Assuming $p \ge 2$:
 \begin{align*}
    P(R^2_i > k_i) &= \frac{1}{\text{B}(p/2,(b_i-p-p_0)/2)}  \int_{k_i}^1 x^{p/2 -1} (1-x)^{(b_i-p-p_0)/2-1} \, \text{d} x \\
    &\le_{(\text{if } p \ge 2)}  \frac{1}{\text{B}(p/2,(b_i-p-p_0)/2)}  \int_{k_i}^1 (1-x)^{(b_i-p-p_0)/2-1} \, \text{d} x \\ 
    &= \frac{ 2 (1-k_i)^{(b_i-p-p_0)/2}}{(b_i-p-p_0)  \,  \text{B}(p/2,(b_i-p-p_0)/2) }.
\end{align*}
If $p = 1$, using the Cauchy-Schwarz inequality:
\begin{align*}
    P(R^2_i > k_i) &= \frac{1}{\text{B}(1/2,(b_i-p_0-1)/2)}  \int_{k_i}^1 x^{-1/2} (1-x)^{(b_i-p_0-1)/2-1} \, \text{d} x \\
    &\le \frac{1}{\text{B}(1/2,(b_i-p_0-1)/2)}   \left( \int_{k_i}^1 x^{-1} \, \mathrm{d} x \right)^{1/2} \, \left( \int_{k_i}^1 (1-x)^{b_i-p_0-3} \, \mathrm{d} x \right)^{1/2} \\
    &= \frac{1}{\text{B}(1/2,(b_i-p_0-1)/2)}  \left(\frac{(1-k_i)^{b_i-p_0-2} \log(1/k_i)}{b_i-p_0-2}\right)^{1/2}, 
\end{align*}
as required.
\end{proof}
\end{proposition}

\subsection{Proof of Proposition 1 in main text}

First of all, note that
$$
\log \tilde{B}_{10} = \sum_{i = 1}^M \log B^c_{10,i}/M + \eta, \, \, \, \log \tilde{I}_{10} = \sum_{i = 1}^M \log I^c_{10,i}/M + \eta,
$$
and recall that $\log {B}^\ast_{10} = (\log \tilde{B}_{10} \vee L) \wedge U$ and $\log {I}^\ast_{10} = (\log \tilde{B}_{10} \vee L) \wedge U.$ Given Assumption 3, the censoring is asymptotically irrelevant. In other words, for any fixed $a \in \mathbb{R}$, there exists a finite $N$ such that $n \ge N$ implies that $P(\log B^\ast_{10} > a) = P(\log \tilde{B}_{10} > a)$ and  $P(\log I^\ast_{10} > a) = P(\log \tilde{I}_{10} > a).$ For that reason, we show consistency by studying the asymptotic behavior of $\tilde{B}_{10}$ and $\tilde{I}_{10}$. Asymptotically, the effect of $\eta$ is irrelevant as well since, by Assumption 4, $\eta \rightarrow_P 0$.

It remains to show that averages of censored logarithms of Bayes factors and information criteria are consistent. We prove consistency by cases. First, we show consistency when $H_0$ is true. Then, we show consistency when $H_1$ is true.

\textbf{Case $H_0$ is true:} We prove consistency for Bayes factors first. The proof for information criteria is very similar.

Using the inequality $(1+x/n)^n \le \exp(x)$, we can bound the non-private Bayes factor as follows:
\begin{align*}
B_{10,i} &= \int_0^\infty (g_i+1)^{-p/2} \left[1 + \frac{g_i R_i^2}{1+g_i(1-R^2_i)} \right]^{(b_i-p_0)/2} \, \pi_i(g_i) \, \nu(\mathrm{d} g_i) \\
&\le \exp \left( \frac{b_i-p_0}{2} \frac{R^2_i}{1-R^2_i} \right) \int_0^\infty (g_i+1)^{-p/2} \, \pi_i(g_i) \, \nu(\mathrm{d} g_i)  \\ &= \exp \left( \frac{b_i-p_0}{2} \frac{R^2_i}{1-R^2_i} \right) \, I_{\pi_i}.
\end{align*}
Let $a \in \mathbb{R}$ be a fixed constant. By Assumption 3, there exists a finite $N$ so that $n \ge N$ implies $L  < a  < U$. Using a union bound, for $n \ge N$:
\begin{align*}
      P \left( \frac{1}{M} \sum_{i=1}^M \log B^c_{10,i} > a \right) &\le \sum_{i=1}^M P( \log B^c_{10,i} > a) \\ 
      &\le \sum_{i=1}^M P(R^2_i > k_i),
\end{align*}
where $k_i = t_i/(1+t_i)$ for $t_i = 2 [a - \log I_{\pi_i}]/(b_i-p_0)$, which is positive for $b_i$ large enough since $\lim_{n \rightarrow \infty} \sup_{i \in 1:M} \int_{0}^\infty  b_i^{p/2} \, (g_i+1)^{-p/2} \pi_i(g_i) \, \nu(\mathrm{d} g_i) < \infty$ by Assumption 6.

Let $p \ge 2$. Under $H_0$, $R^2_i \sim \mathrm{Beta}(p/2, (b_i - p - p_0)/2)$. We show that $\sum_{i=1}^M P(R^2_i > k_i)$ goes to zero with the help of Proposition~\ref{prop:tailR2} and the inequality for the beta function  $1/\mathrm{B}(x,y) \le (x+y)^{x+y-1}/(x^{x-1} y^{y-1})$ (see e.g. Equation 4 in \cite{Greni2015InequalitiesFT}). We also use the inequality $(1+x/k)^{-k} \le \exp(-x)(1-x^2/k)^{-1}$ for $x =  b_i (a - \log I_{\pi_i})/(b_i-p_0)$ and $k = b_i/2 $, which is valid for $|x| \le k$. The condition is satisfied when the sample size is large enough by Assumption 6.  

Putting it all together:
\begin{align*}
\sum_{i=1}^M P(R^2_i > k_i)  &\le  \sum_{i=1}^M \frac{ 2 (1-k_i)^{(b_i-p-p_0)/2}}{(b_i-p-p_0)  \,  \text{B}(p/2,(b_i-p-p_0)/2) } \\
&\lesssim  \sum_{i=1}^M { (b_i-p-p_0)^{p/2-1} (1-k_i)^{b_i/2}}  \\
&=\sum_{i=1}^M \frac{ (b_i-p-p_0)^{p/2-1} }{[ 1+2[a-\log I_{\pi_i}]/(b_i-p_0)]^{b_i/2}} \\
&\lesssim \sum_{i=1}^M { (b_i-p-p_0)^{p/2-1} } I_{\pi_i} \\
&\lesssim \sup_{i \in 1:M} \frac{M}{b_i-p-p_0} \,  b_i^{p/2} I_{\pi_i} \\
&\rightarrow 0.
\end{align*}
In the last steps, we used Assumptions 2 and 6. We can use essentially the same argument to show consistency for information criteria under $H_0$ when $p \ge 2$. In this case, $k_i = 1 - b_i^{-\rho_i/b_i} \exp(-2a/b_i)$ and
\begin{align*}
   P\left( \sum_{i=1}^M \log I^c_{10,i}/M > a \right) &\lesssim \sum_{i=1}^M P(R^2_i > k_i) \\
   &\lesssim \sum_{i=1}^M b_i^{(p-\rho_i)/2-1}  \\
   &\lesssim \sup_{i \in 1:M} \frac{M}{b_i} \, b_i^{(p-\rho_i)/2} \\
   &\rightarrow_{\rho_i \ge p} 0.
\end{align*}

Let $p = 1$. We can use the same strategy we used for $p \ge 2$, but with the appropriate bounds and assumptions. Once again, we can use the inequalities $1/\mathrm{B}(x,y) \le (x+y)^{x+y-1}/(x^{x-1} y^{y-1})$, Proposition~\ref{prop:tailR2}, and assumptions as needed. 

More explicitly: 
\begin{align*}
\sum_{i=1}^M P(R^2_i > k_i)  &\le  \sum_{i=1}^M \frac{1}{\text{B}(1/2,(b_i-1-p_0)/2)}  \left(\frac{(1-k_i)^{b_i-p_0-2} \log(1/k_i)}{b_i-p_0-2}\right)^{1/2} \\
&\lesssim  \sum_{i=1}^M {  \left[{(1-k_i)^{b_i-p_0-2} \log(1/k_i)} \right]^{1/2} } \\  
&\lesssim  \sum_{i=1}^M  I_{\pi_i}  \log(1/k_i)^{1/2} \\
&\lesssim \sum_{i=1}^M  I_{\pi_i}  [\log(b_i/\log b_i) ]^{1/2} \\
&\lesssim \sum_{i=1}^M  b_i^{-1/2} [\log(b_i/\log b_i) ]^{1/2} \\
&\lesssim M \sup_{i \in 1:M}  \sqrt{\log(b_i/\log b_i)/b_i } \\
&\rightarrow 0,
\end{align*}
as required. Using the same argument, we can prove consistency under $H_0$ for information criteria when $p =1$:
\begin{align*}
    P\left( \sum_{i=1}^M \log I^c_{10, i}/M > a \right) &\lesssim \sum_{i=1}^M P(R^2_i > k_i) \\
    &\lesssim \sum_{i=1}^M (1-k_i)^{(b_i-3)/2} \, \log(1/k_i)^{1/2} \\
    &\lesssim_{\rho_i \ge 1} b_i^{-\rho_i - \rho_i/b_i} \, \sqrt{\log b_i} \\
    &\lesssim \frac{M \sqrt{\log b_i}}{b_i} \, b_i^{1-\rho_i(1 - 1/b_i)} \\
    &\rightarrow 0,
\end{align*}
which competes the proof.

\textbf{Case $H_1$ is true:}  First, we show that
$$
\lim_{n\rightarrow\infty} P\left( \forall i \in 1:M , R^2_i \ge b_i^{-0.5} \right) = 1.
$$
Under $H_1$, $R^2_i \stackrel{\text{ind.}}{\sim} \mathrm{noncentralBeta}(p/2, (b_i-p-p_0)/2, \lambda_i)$, where $\lambda_i = \beta' X_i' X_i \beta/\sigma^2$ is the noncentrality parameter of the noncentral beta distribution (type I). In other words, we can write
$$
R^2_i = \frac{X_i}{X_i + Y_i}; \, \, X_i \sim \chi^2_{p}(\lambda_i) \perp Y_i \sim \chi^2_{b_i-p-p_0}.
$$
Let $a_i = b_i^{-0.5}/(1-b_i^{-0.5})$. Then,
$$
R^2_i > b_i^{-0.5} \Leftrightarrow \frac{X_i}{Y_i} > a_i
$$
Let $k_i = a_i [(b_i-p-p_0)+2\sqrt{(b_i-p-p_0)b_i^{0.5}} + 2 b_i^{0.5}]$. Then,
$$
P(X_i/ Y_i > a_i) \ge P(X_i > k_i) P(Y_i < k_i/a_i).
$$
Provided that $\lambda_i + p > k_i$, Proposition~\ref{prop:tailbounds} ensures that 
$$
P(X_i > k_i) \ge 1 - \exp \left( - \frac{(p+\lambda_i -k_i)^2}{4 (p+2 \lambda_i)} \right).
$$
The condition $\lambda_i + p > k_i$ is satisfied for a large enough sample size for all $i \in \{1,  2, \, ... \, , M\}$ under Assumption 5. On the other hand, also by Proposition~\ref{prop:tailbounds},
$$
P(Y_i < k_i/a_i)  \ge 1 - \exp ( - b_i^{0.5}),
$$
Thus, under Assumption 2:
\begin{align*}
P \{ \forall i \in 1:M , R^2_i \ge b_i^{-0.5} \} &\ge  \prod_{i = 1}^M \{ P(X_i > k_i) P(Y_i < k_i/a_i) \} \\
&\ge \inf_{i \in 1:M}  \left\{ \left[1 - \exp \left( - \frac{(p+\lambda_i -k_i)^2}{4 (p+2 \lambda_i)} \right)\right]^{b_i} \left[ 1 - \exp ( - b_i^{0.5}) \right]^{b_i} \right\}^{M/b_i} \\
&\rightarrow 1,
\end{align*}
as required. 

Conditioning on the high probability event $\{ \forall i \in 1:M , R^2_i \ge b_i^{-0.5} \}$:
\begin{align*}
\frac{1}{M} \sum_{i=1}^M \log B^c_{10,i} &\ge \frac{1}{M} \sum_{i=1}^M \log \int_0^\infty \frac{(g_i+1)^{(b_i-p-p_0)/2}}{[1+g_i (1- b_i^{-0.5})]^{(b_i-p_0)/2}} \pi_i(g_i) \, \nu(\mathrm{d} g_i) \\
&\gtrsim \log  \inf_{i \in 1:M} \left[\frac{b_i +1}{b_i(1-b_i^{-0.5}) +1}\right]^{(b_i-p_0)/2}   \int_{b_i}^\infty (g_i+1)^{-p/2}   \pi_{i}(g_i) \, \nu( \mathrm{d} g_i ) \\
&\rightarrow \infty
\end{align*}
and, with high probability, 
\begin{align*}
    \frac{1}{M} \sum_{i=1}^M \log I^c_{10, i} &\gtrsim \inf_{i \in 1:M} \log \left( \frac{b_i^{-\rho_i/2}}{(1- b_i^{-0.5})^{b_i/2}} \right) \\
    &\rightarrow \infty,
\end{align*}
which completes the proof.

\subsection{Auxiliary results for Proposition 2}

\begin{proposition}  (Hanson-Wright inequality; see e.g. \cite{rudelson2013hanson}). Let $W$ be a random vector in $\mathbb{R}^n$. Let $\{W_i\}_{i=1}^n$ be the components of $W$. Let $\mathbb{E}(W_i) = 0$ and $0 < K < \infty$ be a constant so that $\mathbb{E}[\exp(W_i^2/K^2)] < 2$. Let $A$ be a $n \times n$ real matrix. Then, for a constant $c > 0$, 
$$
P [ | W' A W - \mathbb{E}(W'A W) | > t ] \le 2 \exp\left\{ - c \frac{t}{K^2} \min \left(  \frac{t}{K^2 \lVert A \rVert_F^2} , \frac{1}{ \lVert A \rVert_{op}} \right)  \right\},
$$
where $\lVert A \rVert_F = \sqrt{\mathrm{tr}(A'A)}$ is the Frobenius norm of $A$ and $\lVert A \rVert_{op} = \sup_{\lVert x \rVert \neq 0} \lVert A x \rVert/ \lVert x  \rVert$. 
\end{proposition}

\begin{proposition} (Convergence of $G^\ast_r$ and $G^{\ast \ast}_r$) \label{prop:convG}
 Under the assumptions of Proposition 2 in the main text, $G^{\ast}_r/n \rightarrow_P  \overline{G}_{\infty}$ and $G^{\ast \ast}_r/n \rightarrow_P  \overline{G}_{\infty},$ a constant symmetric matrix. 
\end{proposition}

\begin{proof} Let $G^{\ast}/n = G/n + E/n$. The entries $E_{ij}/n$ converge to zero in probability by Assumption 3.

We show that $G/n \rightarrow_P \overline{G}_{\infty}$ with  $\overline{G}_{\infty}$ symmetric. Recall that
$$
G/n = \begin{bmatrix} V'V/n & V'Z/n \\ 
Z'V/n & Z'Z/n
\end{bmatrix}.
$$
We establish the convergence of $G/n$ to $\overline{G}_{\infty}$ by establishing the convergence of the blocks. All the probability statements, expectations, and variances in this proof are given $X_0$ and $V$. 

By Assumption 5, $\lim_{n \rightarrow \infty} V'V/n = S_1$. Then,
\begin{align*}
    \lim_{n \rightarrow \infty} \mathbb{E}(V'Z/n) &= \lim_{n \rightarrow \infty} V'V_T \beta_T/n = S_2 \beta_T \\
    \lim_{n \rightarrow \infty} \mathrm{Var}(V'Z/n) &= \lim_{n \rightarrow \infty} \frac{\sigma^2_T}{n^2} V'V = 0_{p \times p} \\
    V'Z/n &\rightarrow_{P} S_2 \beta_T,
\end{align*}
where $S_2$ is a submatrix of $S_1$.

Let us study the asymptotic behavior of $Z'Z/n$. We can write 
\begin{align*}
Z'Z/n &= Q/n + 2 \beta_T' V_T' Y/n - \beta_T' V_T' V_T \beta_T/n, \\
Q &= [Y - \mathbb{E}(Y)]'(I_n - P_{X_0})[Y - \mathbb{E}(Y)].
\end{align*}
On the one hand,
\begin{align*}
\lim_{n \rightarrow \infty} \mathbb{E}(2 \beta_T' V_T' Y/n - \beta_T' V_T' V_T \beta_T/n) &= \lim_{n \rightarrow \infty}  \beta_T' V_T' V_T \beta_T/n = \beta_T' S_3 \beta_T \\
\lim_{n \rightarrow \infty} \mathrm{Var}(2 \beta_T' V_T' Y/n - \beta_T' V_T' V_T \beta_T/n) &= \lim_{n \rightarrow \infty} \frac{4 \sigma^2_T}{n^2} \beta_T' V_T' V_T \beta_T  = 0 \\
2 \beta_T' V_T' Y/n - \beta_T' V_T' V_T \beta_T/n &\rightarrow_P \beta_T' S_3 \beta_T,
\end{align*}
where $S_3$ is a submatrix of $S_1$. 

We can show that $Q/(n-p_0) \rightarrow_P \sigma^2_T$ with the Hanson-Wright inequality, which in turn implies that $Q/n \rightarrow_P \sigma^2_T$. We have $\mathbb{E}[Q/(n-p_0)] = \sigma^2_T$. Define $W = Y -  \mathbb{E}(Y)$. Then, $\mathbb{E}(W) = 0_n$ and $W_i^2 $ are uniformly bounded since both $Y$ and $\mathbb{E}(Y)$ are by Assumption 1. Therefore, we can pick a finite constant $K$ satisfying $\mathbb{E}[\exp(W_i^2/K^2)] < 2$. 

Applying the Hanson-Wright inequality, for any  given $\epsilon > 0$,
$$
P[ |Q/(n-p_0) - \sigma^2_T| > \epsilon] \le 2 e^{ - 2 c \epsilon \min \left[ {\epsilon}(K^2  \lVert (I - P_{X_0})/(n-p_0) \rVert_F)^{-1},  (\lVert (I - P_{X_0})/(n-p_0) \rVert_{op})^{-1} \right]} \rightarrow_{n \rightarrow \infty} 0,
$$
because $\lVert (I - P_{X_0})/(n-p_0) \rVert_{op} =  \lVert (I - P_{X_0})/(n-p_0) \rVert_F = 1/\sqrt{n-p_0} \rightarrow 0$ as $n \rightarrow \infty$. This implies $Q/n \rightarrow_P \sigma^2_T$ and $Z'Z/n \rightarrow_P \sigma^2_T + \beta_T' S_3 \beta_T$. 

We have shown that
$$
G/n \rightarrow_P G/n = \begin{bmatrix} S_1 & S_2 \beta_T \\ 
\beta_T' S_2' & \sigma^2_T + \beta_T' S_3 \beta_T
\end{bmatrix} = \overline{G}_{\infty}.
$$
Therefore, we have $G^\ast/n \rightarrow_P \overline{G}_{\infty}$.

In the case of the non-thresholded matrix $G^\ast_r$, we have established that $G^\ast_r/n = G^\ast/n + r/n I_{p+1} \rightarrow_P \overline{G}_{\infty}$ since $r/n \rightarrow 0$ by Assumption 4. In the case of $G^{\ast \ast}_r$, there is an indicator that can hard-threshold off-diagonal elements. Let $G^{\ast \ast}_{ij}/n =  G^\ast_{ij}/n \, \mathbbm{1}(i = j \, \mathrm{or} \, |G^{\ast}_{ij}|/n \ge e_{\lambda}/n)$ be the $(i,j)$-th entry of $G^{\ast \ast}$. On the one hand, $\lim_{n \rightarrow \infty} e_{\lambda}/n = 0$ because the variance of $E_{ij}$ is finite and does not depend on $n$.  For $i  = j$ the indicator is equal to 1. For $i \neq j$:
\begin{align*}
    E [ \mathbbm{1}( |G^{\ast}_{ij}|/n \ge e_{\lambda}/n) ] &= P [|G^{\ast}_{ij}|/n \ge e_{\lambda}/n] \rightarrow 1 \\
    \mathrm{Var}[\mathbbm{1}( |G^{\ast}_{ij}|/n \ge e_{\lambda}/n] &=  P [|G^{\ast}_{ij}|/n \ge e_{\lambda}/n] -  P [|G^{\ast}_{ij}|/n \ge e_{\lambda}/n]^2 \rightarrow 0 \\
    \mathbbm{1}( |G^{\ast}_{ij}|/n \ge e_{\lambda}/n) &\rightarrow_P 1.
\end{align*}
By Slutzky's lemma, we have that $G^{\ast \ast}_{ij}/n \rightarrow \overline{G}_{\infty, ij}$ for all $i,j$. This is enough to show that $G^\ast/n \rightarrow \overline{G}_{\infty}$. Finally, since $G^{\ast \ast}_r = G^{\ast \ast} + r I_{p+1}$ and $r/n \rightarrow 0$ by Assumption 4, we have $G^{\ast \ast}_r/n \rightarrow_P \overline{G}_{\infty}$, as required. 

\end{proof}

\begin{proposition} \label{prop:noisyR2} (Convergence of noisy $R^2$) Given $G^\ast_r$ or $G^{\ast \ast}_r$ and a model-indexing vector $\gamma \in \{0,1\}^p$, we can construct a differentially private version of $R^2_\gamma$, denoted $R^{2, \ast}_{\gamma}$. Let $T \in \{0, 1\}^p$ be the index of the true model. Under the regularity conditions stated in Proposition~{\ref{prop:convG}} and assuming that $Z'Z$ is not equal to zero almost surely, the following are true:
\begin{enumerate}
    \item If $\gamma$ does not nest $T$ (i.e. if there exists $i$ such that $T_i = 1$ and $\gamma_i = 0$), then $R^{2, \ast}_{\gamma} \rightarrow R^{2}_{\gamma, \infty}$ and  $R^{2, \ast}_{T} \rightarrow R^{2}_{T, \infty}$ with $R^{2}_{\gamma, \infty} < R^{2}_{T, \infty}$. 
    \item If the $T = 0_p$ is the null model, then $n R^{2, \ast}_{\gamma}$ is in $O_p(1)$.
    \item If $\gamma$ nests $T$ (i.e. if $T_i = 1$ implies $\gamma_i = 1$), then $[(1-R^{2, \ast}_T)/(1-R^{2,\ast}_\gamma)]^{(n-p_0)/2}$ is in $O_p(1)$.
\end{enumerate}
\end{proposition} 

\begin{proof} We prove the three statements separately.

\textbf{Proof of statement 1.} The proofs of $R^2_\gamma \rightarrow_P R^2_{\gamma, \infty}$ and $R^{2, \ast}_{\gamma} \rightarrow_P R^2_{\gamma, \infty}$ are a direct consequence of Proposition~{\ref{prop:convG}}. Note that 
$$
R^2_{\gamma} = \frac{Z'V_\gamma (V_\gamma'V_\gamma)^{-1} V_\gamma' Z}{Z'Z} = \frac{Z' P_{V_{\gamma}} Z/n}{Z'Z/n}.
$$
On the one hand, $n/Z'Z \rightarrow_P 1/(\sigma^2_T + \beta_T' S_3 \beta_T)$. Then, $V_\gamma' Z$ are subvectors of $V'Z$, so $V_{\gamma}'Z/n$ converges to a subvector of $S_2 \beta_T$. Similarly, $V_{\gamma}' V_{\gamma}/n$ converges to a submatrix of $S_1$ and, since we assume that $V$ is full-rank, $V_{\gamma}'V_{\gamma}$ is invertible for any $\gamma$. This implies that $R^2_{\gamma}$ converges in probability to some constant $R^{2}_{\gamma, \infty}$, which has to be between zero and one because $V_\gamma (V_\gamma'V_\gamma)^{-1} V_\gamma$ is a projection matrix and $Z' P_X Z \le Z'Z$ for any projection matrix $P_X$.

The convergence of the noisy $R^{2, \ast}_{\gamma}$ to $R^{2}_{\gamma, \infty}$ can be established after noting that $R^{2, \ast}_{\gamma}$ can be constructed by taking submatrices of $G^\ast_r$ or $G^{\ast \ast}_r$ and multiplying and dividing terms as needed. We can invoke Proposition~\ref{prop:convG} and Slutzky's lemma and conclude that $R^{2, \ast}_{\gamma} \rightarrow_P R^2_{\gamma, \infty}$, as required. 

It remains to show that if $\gamma$ does not nest the true model $T$, $R^{2}_{\gamma, \infty} < R^2_{T, \infty}$. To see this, note that  
$$
\lim_{n \rightarrow \infty} \mathbb{E}(Z' P_{V_{\gamma}} Z/n) = \lim_{n \rightarrow \infty} \beta_T' X_T' P_{V_{\gamma}} X_T \beta_T/n, \, \, \lim_{n \rightarrow \infty} \mathrm{Var}(Z'P_{V_{\gamma}} Z/n) = 0,
$$
and $\beta_T' X_T' P_{V_{\gamma}} X_T \beta_T < \beta_T' X_T' X_T \beta_T  = \beta_T' X_T' P_{V_{T}} X_T \beta_T$, which implies $R^2_{\gamma, \infty} < R^2_{T, \infty}.$

\textbf{Proof of statement 2.} If $T$ is the null model, we show that $n R^{2, \ast}_{\gamma}$ is in $O_p(1)$. It is useful to write 
$$
n R^{2, \ast}_{\gamma} = \frac{\frac{1}{\sqrt{n}} (Z'V)^{\ast}_{\gamma} \left[ \frac{(V'V)^{\ast}_{\gamma}}{n}\right]^{-1} (V'Z)^{\ast}_{\gamma} \frac{1}{\sqrt{n}}}{(Z'Z)^{\ast}/n}.
$$
By Proposition~\ref{prop:convG}, we know that the denominator converges in probability to a constant. It is enough to show that the numerator is in $O_p(1)$. The matrix $\left[ {(V'V)^{\ast}_{\gamma}}/{n}\right]^{-1}$ is in $O_p(1)$ by Proposition~\ref{prop:convG}. It remains to show that $(V'Z)^{\ast}_{\gamma}/{\sqrt{n}}$ is in $O_p(1)$. Let $w^\ast = (V'Z)  + E_{2}$ and $k_{\lambda}$ the appropriate threshold for the off-diagonal elements of $G^{\ast \ast}$. Then,
\begin{align*}
    \lVert \frac{1}{\sqrt{n}} (V'Z)^\ast_{\gamma} \rVert^2 &\le \frac{1}{n} \lVert  w^\ast \, \mathbbm{1}(|w^\ast| > k_{\lambda}) \rVert^2 \\
    &\le \frac{1}{n} \lVert  w^\ast \rVert^2 \\
    &\le \frac{1}{n} \lVert V'Z \rVert^2 + \frac{1}{n} \lVert E_{2} \rVert^{2}.
\end{align*}
It  suffices to show that both the error $E_{2}/\sqrt{n}$ and the non-private $V'Z/\sqrt{n}$ are in $O_p(1)$. Each of the $E_{ij}$ in $E$ has a finite variance that does not depend on $n$. Therefore, each $E_{ij}/\sqrt{n}$ has a variance that goes to zero and, since $E_2$ has $p$ elements, $\lVert E_{2}/\sqrt{n} \rVert$ is in $O_p(1)$. It only remains to show that $V'Z/\sqrt{n}$ is in $O_p(1)$, which we show using the Hanson-Wright inequality. 

First, note that $
\lVert V'Z/\sqrt{n} \rVert^{2} = Z' VV' Z/n$. Then, since we are assuming that the true model is the null model:
$$
\mathbb{E}(Z'VV'Z/n) = \frac{\sigma^2}{n} \mathrm{tr}(V'V) \rightarrow_{n \rightarrow \infty} c > 0.
$$
The limit is a positive constant because $V'V/n$ converges to a symmetric positive-definite matrix by Assumption 5, which also implies $\lim_{n \rightarrow \infty} \lVert VV'/n \rVert_{F} < \infty$ and  $\lim_{n \rightarrow \infty} \lVert VV'/n \rVert_{op} < \infty.$ From here, we can apply the Hanson-Wright inequality to establish that $|Z'VV'Z/n - \mathbb{E}(Z'VV'Z/n)|$ is in $O_p(1)$ and, since $ \mathbb{E}(Z'VV'Z/n)$ converges to a constant, we conclude that $Z'VV'Z/n$ is in $O_p(1)$, as required.

\textbf{Proof of statement 3.} {Finally, we show that if $\gamma$ nests $T$, $[(1-R^{2, \ast}_T)/(1-R^{2,\ast}_\gamma)]^{(n-p_0)/2}$ is in $O_p(1)$. Consider the non-private $[(1-R^{2}_T)/(1-R^{2}_\gamma)]^{(n-p_0)/2}$}. We can write 
\begin{align*}
\left[\frac{1-R^{2}_T}{1-R^{2}_\gamma}\right]^{(n-p_0)/2} &= \left[1 + \frac{2}{n-p_0}  \frac{Z'(P_{V_{\gamma}} - P_{V_T} )Z}{2 Z'(I_n - P_{V
_{\gamma}}) Z/(n-p_0)} \right]^{(n-p_0)/2} \\
&\le \exp\left(\frac{Z'(P_{V_{\gamma}} - P_{V_T} )Z}{2 Z'(I_n - P_{V
_{\gamma}}) Z/(n-p_0)} \right).
\end{align*}
The denominator $Z'(I_n - P_{V
_{\gamma}}) Z/(n-p_0)$ converges to a constant. This fact follows directly given the asymptotic behavior of $Z' P_{V_{\gamma}} Z/n$ and $Z'Z/n$ we just described in the proof of the first statement. The same is true for the private version of the statistic, using the argument we used for $R^{2, \ast}_{\gamma}$. 

The numerator $Z'(P_{V_{\gamma}} - P_{V_T} )Z$ is in $O_p(1)$, which can be shown using the Hanson-Wright inequality. The expectation is $\mathbb{E}[Z'(P_{V_{\gamma}} - P_{V_T} )Z] = \sigma^2 (|\gamma| - |T|)$. Let $Q_{\gamma} = Z'(P_{V_{\gamma}} - P_{V_T} )Z - \sigma^2 (|\gamma| - |T|)$. Applying the Hanson-Wright inequality
$$
P( Q_{\gamma} > M) \le 2 \exp\left\{ - c \frac{M}{K^2} \min\left( \frac{M}{K^2 (|\gamma| - |T|)}, 1 \right) \right\},
$$
where $K$ is a constant that can be chosen in a similar way as we did in Proposition~\ref{prop:convG}. The right-hand side can be made arbitrarily close to zero by increasing $M$, which establishes that $Q_{\gamma}$ and $Z'(P_{V_{\gamma}} - P_{V_T}) Z$ are in $O_p(1)$. The same is true for the private version of the statistic, using an argument which is similar to the one we used for showing that $n R^{2, \ast}_{\gamma}$ is in $O_p(1)$ when the true model is the null model. Both  $(V'Z)^\ast_{\gamma}/\sqrt{n}$ and $(V'Z)^\ast_{T}/\sqrt{n}$ converge to their private versions $(V'Z)_{\gamma}/\sqrt{n}$ and $(V'Z)_{T}/\sqrt{n}$ because the error goes to zero and the indicator converges (see Proposition~\ref{prop:convG} and the earlier proof for $n R^{2, \ast}_{\gamma}$ for more detailed versions of these arguments). The private $(V'V)_{\gamma}^{\ast}$ also converge to their non-private versions (see e.g. Proposition~\ref{prop:convG}). Therefore, the private $(Z'(P_{V_{\gamma}} - P_{V_T})Z)^{\ast}$ has the same asymptotic behavior as the private $Z'(P_{V_{\gamma}} - P_{V_T})Z,$ which we have shown to be in $O_p(1)$. This completes the proof.


\end{proof}

\subsection{Proof of Proposition 2 in main text} 

We consider two cases: one where the true model is the null model and another one where the true model is not the null model.

\textbf{True model is the null model:} Let $\gamma$ be a model that is not the null model. Then,  we have
$$
B^{\ast}_{\gamma 0} \le \exp\left( \frac{n-p_0}{2} \frac{R^{2,\ast}_{\gamma}}{1-R^{2,\ast}_{\gamma}} \right) \int_{0}^\infty (g+1)^{-|\gamma|/2} \, \pi(g) \, \, \nu(\mathrm{d} g).
$$
The integral converges to zero by Assumption 6 and the exponential term is in $O_p(1)$ because, by Proposition~\ref{prop:noisyR2}, we know that $n R^{2,\ast}_{\gamma}$ is in $O_p(1)$.  Therefore, for any model $\gamma$ which is not the null model, $B^\ast_{\gamma0} \rightarrow_P 0$. A similar argument works for information criteria. In such case,
$$
I^\ast_{\gamma 0} \le n^{-\rho_{|\gamma|}/2} \exp\left( \frac{n}{2} \frac{R^{2,\ast}_{\gamma}}{1-R^{2,\ast}_{\gamma}} \right)\rightarrow_P 0.
$$

\textbf{True model is not the null model:} We study the asymptotic behavior of $B^\ast_{\gamma T}$, where $T$ is the true model and $\gamma$ is a model that is not the true model. We split this case into two subcases: one where $\gamma$ nests the true model, and another one where $\gamma$ does not nest the true model. 

First, note that, since we are working with null-based Bayes factors,
$$
B^{\ast}_{\gamma T} = \frac{B^{\ast}_{\gamma 0} (1 - R^{2,\ast}_T)^{(n-p_0)/2}}{ B^{\ast}_{T0} (1 - R^{2, \ast}_T)^{(n-p_0)/2}}  
$$
we will bound the numerator and denominator separately and put our bounds together. First, we bound the numerator:
\begin{align*}
    B^{\ast}_{\gamma 0} (1 - R^{2,\ast}_T)^{(n-p_0)/2} &= \int_{0}^\infty (g+1)^{-|\gamma|/2} \left[\frac{1 +g (1-R^{2,\ast}_T) -R^{2, \ast}_T}{1 + g(1-R^{2,\ast}_{\gamma})}  \right]^{(n-p_0)/2} \, \pi(g) \, \nu(\mathrm{d} g) \\
    &\le \left( \frac{1-R^{2,\ast}_{T}}{1-R^{2,\ast}_{\gamma}} \right)^{(n-p_0)/2} \, \int_{0}^\infty (g+1)^{-|\gamma|/2} \, \pi(g) \, \nu(\mathrm{d} g).
\end{align*}
Then, we bound the denominator:
\begin{align*}
    B^{\ast}_{T0} (1 - R^{2, \ast}_T)^{(n-p_0)/2} &\ge \int_{n}^{\infty} (g+1)^{-|T|/2} \left[1- \frac{R^{2, \ast}_T}{1 + g(1-R^{2,\ast}_{T})}  \right]^{(n-p_0)/2} \, \pi(g) \, \nu(\mathrm{d} g) \\
    &\ge \left[1- \frac{R^{2, \ast}_T}{1 + n(1-R^{2,\ast}_{T})}  \right]^{(n-p_0)/2} \int_{n}^{\infty} (g+1)^{-|T|/2}  \, \pi(g) \, \nu(\mathrm{d} g).
\end{align*}
Putting the bounds together and using Assumption 6:
\begin{align*}
B^{\ast}_{\gamma T}  &\le \left(\frac{1-R^{2,\ast}_{T}}{1-R^{2,\ast}_{\gamma}}\right)^{(n-p_0)/2}  \left({1- \frac{R^{2, \ast}_T}{1 + n(1-R^{2,\ast}_{T})}}\right)^{-(n-p_0)/2} \, \frac{\int_{0}^\infty (g+1)^{-|\gamma|/2} \, \pi(g) \, \nu(\mathrm{d} g)}{\int_{n}^{\infty} (g+1)^{-|T|/2}  \, \pi(g) \, \nu(\mathrm{d} g)}  \\
&\lesssim n^{(|T|-|\gamma|)/2}\left(\frac{1-R^{2,\ast}_{T}}{1-R^{2,\ast}_{\gamma}}\right)^{(n-p_0)/2}  \exp(R^{2,\ast}_{\gamma}/(1-R^{2,\ast}_{\gamma})).
\end{align*}
When $\gamma$ nests $T$, $n^{(|T|-|\gamma|)/2}$ goes to zero and the remaining terms are in $O_p(1)$ by Proposition~\ref{prop:noisyR2}, so $B^{\ast}_{\gamma T}$ converges to zero in probability. When $\gamma$ does not nest $T$, we show that
$$
A_n = n^{(|T| - |\gamma|)/2}[({1-R^{2,\ast}_{T}})/({1-R^{2,\ast}_{\gamma}})]^{(n-p_0)/2} \rightarrow_P 0
$$ 
Let $R_n = ({1-R^{2,\ast}_{T}})/({1-R^{2,\ast}_{\gamma}})$, which converges in probability to a constant less than one. Taking logarithms
$$
\log \left\{ n^{(|T| - |\gamma|)/2} R_n^{(n-p_0)/2} \right\} = \frac{n-p_0}{2}\left[(|T| - |\gamma|) \log n/(n-p_0) + \log R_n \right],
$$
which diverges to $-\infty$ in probability, so $A_n \rightarrow_P 0$  The remaining term in the upper bound for $B^\ast_{\gamma T}$ is in $O_p(1)$. Therefore, we have shown that $B^\ast_{\gamma T} \rightarrow_P 0$. 

The proof for information criteria is essentially the same, but we do not need to bound integrals. Indeed,
$$
I^\ast_{\gamma T} = n^{(\rho_{|T|} - \rho_{|\gamma|})/2} \left(\frac{1-R^{2,\ast}_{T}}{1-R^{2,\ast}_{\gamma}}\right)^{n/2},
$$
and we can use the same arguments we used for $B^\ast_{\gamma T}$ to show that $I^\ast_{\gamma T}$ is consistent.

\subsection{Proof of Proposition 3 in main text}

Let $\mathcal{U} \in \mathbbm{R}^{n \times (p+1)}$ be a full-rank matrix and define $\mathcal{M} = (I_n-{P}_{X_0}) \mathcal{U}$. Given a Gram matrix $\mathcal{G}$, we can generate a synthetic dataset $\mathcal{D} = [\mathcal{V}; \mathcal{Z}]$ with the formula $\mathcal{D} = \mathcal{M} (\mathcal{M}'\mathcal{M})^{-1/2} \mathcal{G}^{1/2}$. In other words, we have $\mathcal{D}'\mathcal{D} = \mathcal{G}$:

$$
\mathcal{D}'\mathcal{D} =  \mathcal{G}^{1/2} \mathcal{M} (\mathcal{M}'\mathcal{M})^{-1/2} \mathcal{M}' \mathcal{M} (\mathcal{M}'\mathcal{M})^{-1/2} \mathcal{G}^{1/2} = \mathcal{G}^{1/2} \mathcal{G}^{1/2} = \mathcal{G}.
$$

The synthetic data are also centered (the same way that $V$ and $Z$ are centered in our construction of $D$). This is true because $\mathcal{D}$ is pre-multiplied by $I_n-P_{X_0}$, so it is orthogonal to the span of $X_0 = 1_n$.

\section{Effects of censoring and data-splitting} \label{suppl:hsb2cens}

In this section, we revisit the High School and Beyond Survey data  set (Section~\ref{subsec:hsb2}) to study the effects of setting different censoring limits. For concreteness, we restrict our attention to the test of $H_{02}$ against $H_{12}$; that is, the hypothesis test where we wish to know if \texttt{read} scores are predictive of \texttt{math} scores when \texttt{science} scores are already a covariate in the model. The results for the test of $H_{01}$ against $H_{11}$ are similar. The simulation setup is the same as described in Section~\ref{suppl:hsb2cens}. The only changes are the censoring limits.

In Figure~\ref{fig:postprobs_cens}, we display the results for the differentially private posterior probability of $H_{12}$ for different values of $\varepsilon$ and $M$. We censor the posterior probability of $H_{12}$ at $[0.35, 0.65]$, $[0.25, 0.75]$, $[0.01, 0.99]$, and $[0.001, 0.999]$. The true, non-private posterior probability of $H_{12}$ is near 1. For all censoring limits, increasing the number of subgroups decreases the variance of the output, but it induces a bias that shrinks the probability to $0.5$. More stringent censoring limits reduce the variance as well, but come at the cost of potential bias: for instance, censoring the posterior probability at $[0.35, 0.65]$ is clearly too stringent, since the true posterior probability is much higher than the upper limit.

In Figure~\ref{fig:LR_cens}, we display the analogous result for likelihood ratios.  In this case, the lower censoring limit is set to $L = 0$, which is a natural lower bound for likelihood ratios. The gray $+$ and $\vartriangle$ are calibrated critical values for rejection of $H_{02}$ at significance at levels 0.01 and 0.05, respectively. The upper censoring limits $U$ are set to 1, 2, 7, and 10. We arrive at the same conclusions we reached with the Bayesian analysis. The true, non-private likelihood ratio is above 10, and if we censor at a much lower value (such as 1 or 2) the test fails to be powerful. The test performs best if the lower bound is 10, in which case the test is quite powerful, especially for $M \ge 5$ and $\delta \ge 0.25$.

\begin{figure}[h!]
    \centering
    \includegraphics[width=0.75\textwidth]{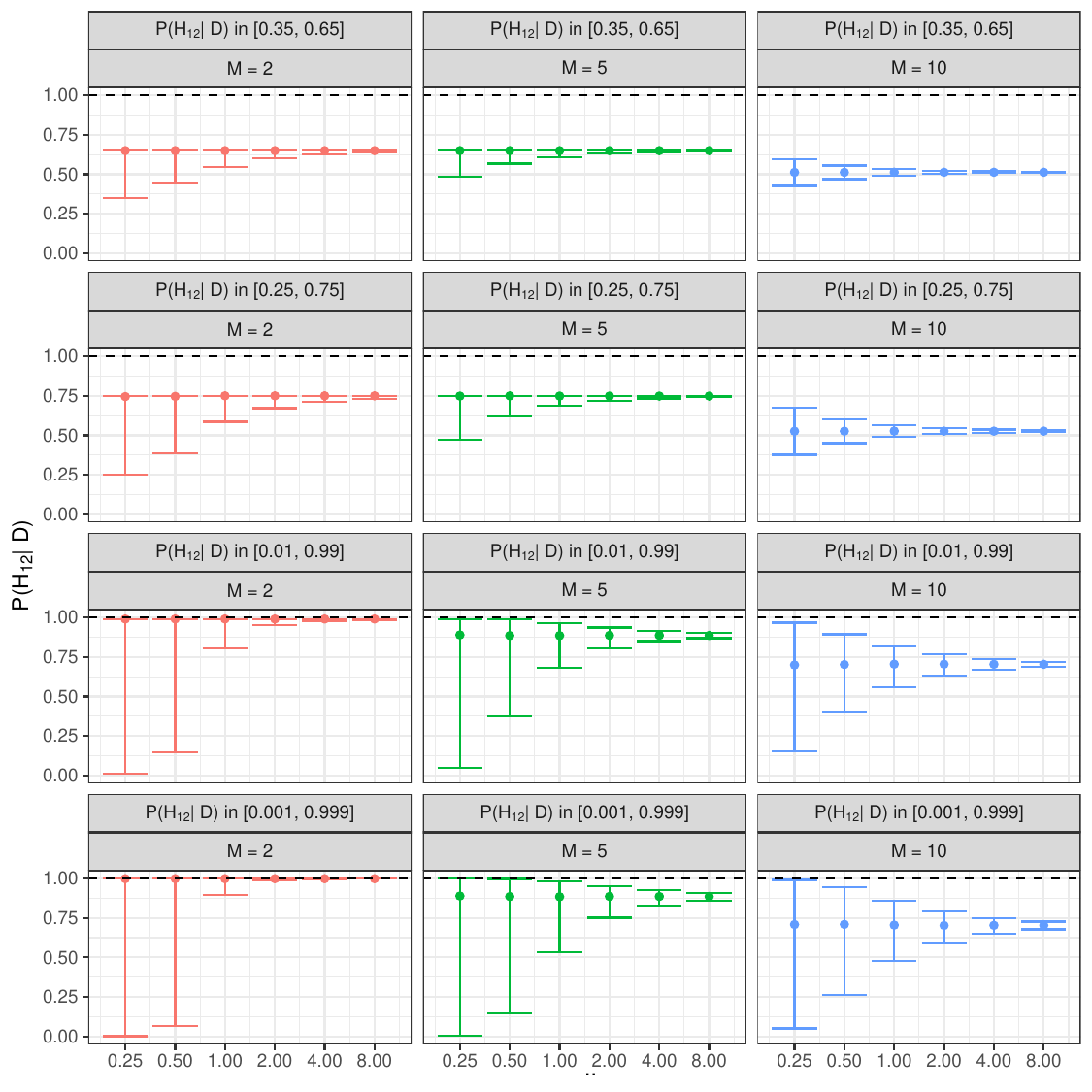}
  \caption{ Distribution of $P^\ast(H_{12} \mid D)$  as a function of $\varepsilon$, $M$, and censoring limits. The lower endpoint of the error bars is the first quartile of the distribution, the midpoint is the median, and the upper endpoint is the third quartile. The dashed lines are the non-private posterior probabilities.}
 \label{fig:postprobs_cens}
\end{figure}

\begin{figure}[h!]
    \centering
    \includegraphics[width=0.75\textwidth]{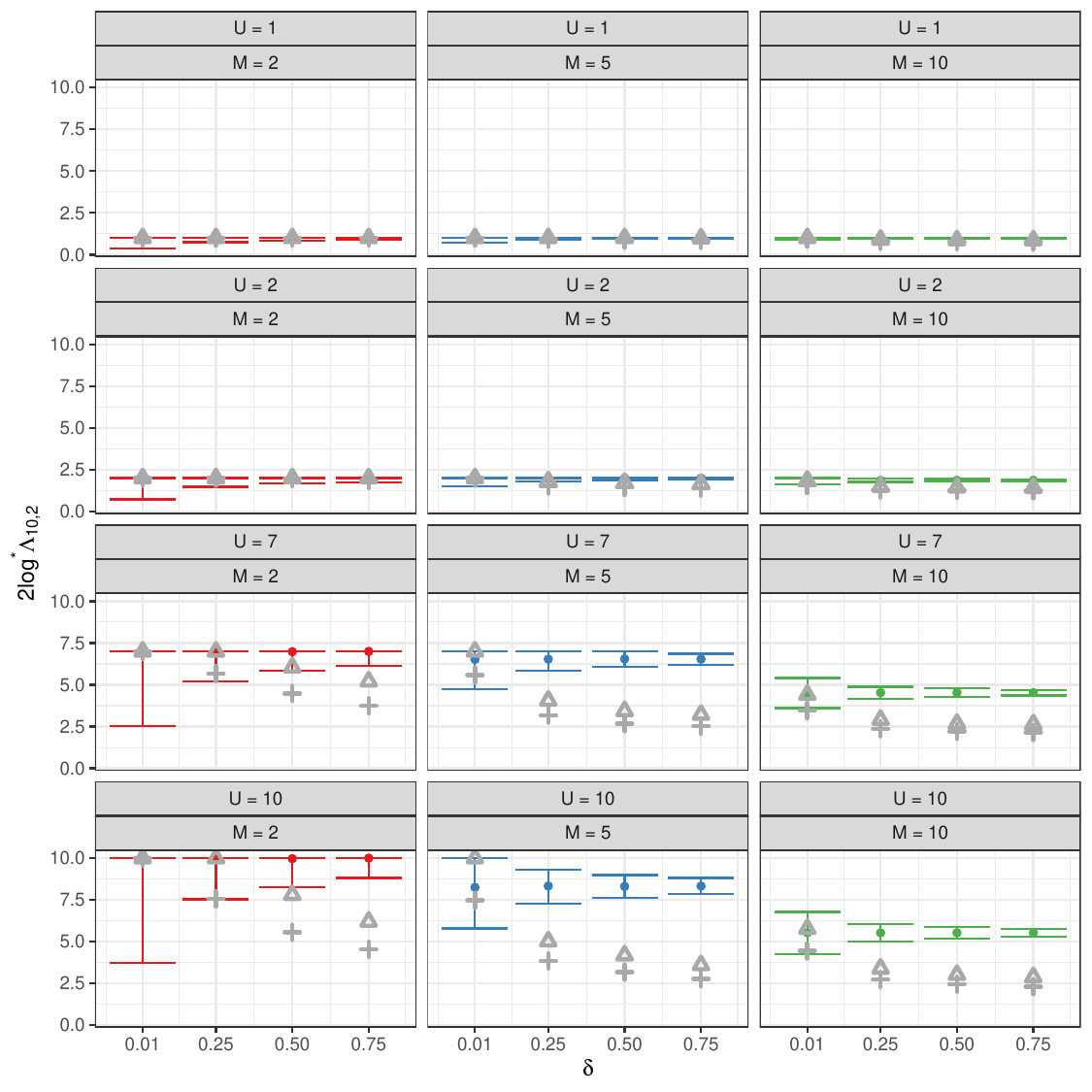} 
    \caption{Distribution of  $2 \log \Lambda^{\ast }_{10, 1}$  and $2 \log \Lambda^{\ast }_{10, 2}$ as a function of $\delta$, $M$, and censoring upper limit $U$. The gray $+$ and $\vartriangle$ are corrected critical values at the 0.01 and 0.05 significance levels, respectively.}
    \label{fig:LR_cens}
\end{figure}

\color{black}

\section{Additional Plots for Simulation Study}\label{suppl:plots}

In this section, we include additional plots for the simulation study in Section~\ref{subsec:empeval}. We show results with BIC combined with least-squares estimates, as well as results for additional values of $\varepsilon$ that were not included in the main text. The interpretation of the plots is the same: thresholded methods perform best when $|T|$ is small, and non-thresholded methods perform best when $|T|$ is large.

\begin{figure}[h!]
  \centering
  \includegraphics[scale= 0.55]{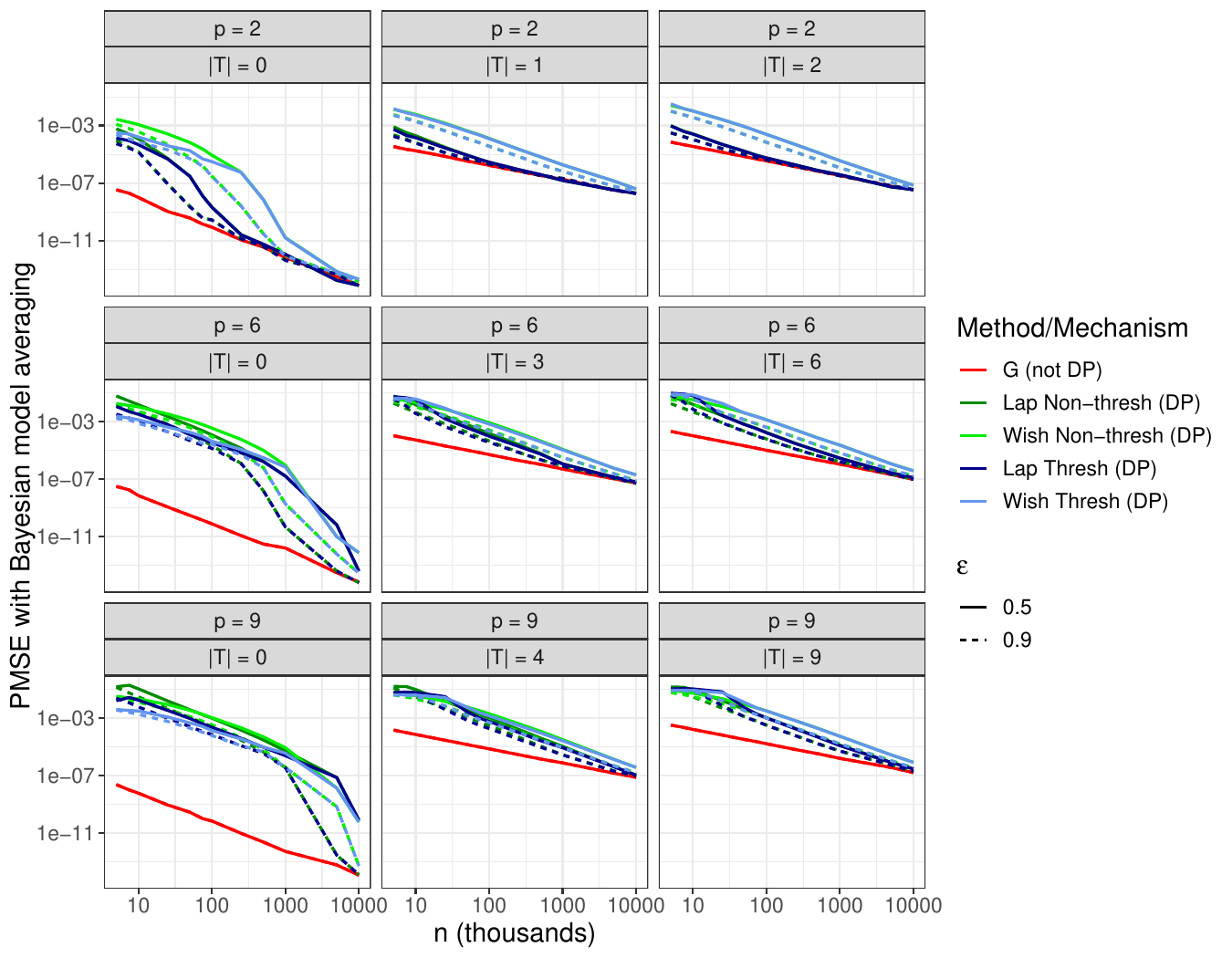}
  \caption{Simulation study: Sample size ($x$-axis) against log(PMSE) ($y$-axis) with BIC prior.}
 \label{fig:log_MSE_simulation_BIC}
\end{figure}

\begin{figure}[h!]
  \centering
  \includegraphics[scale= 0.55]{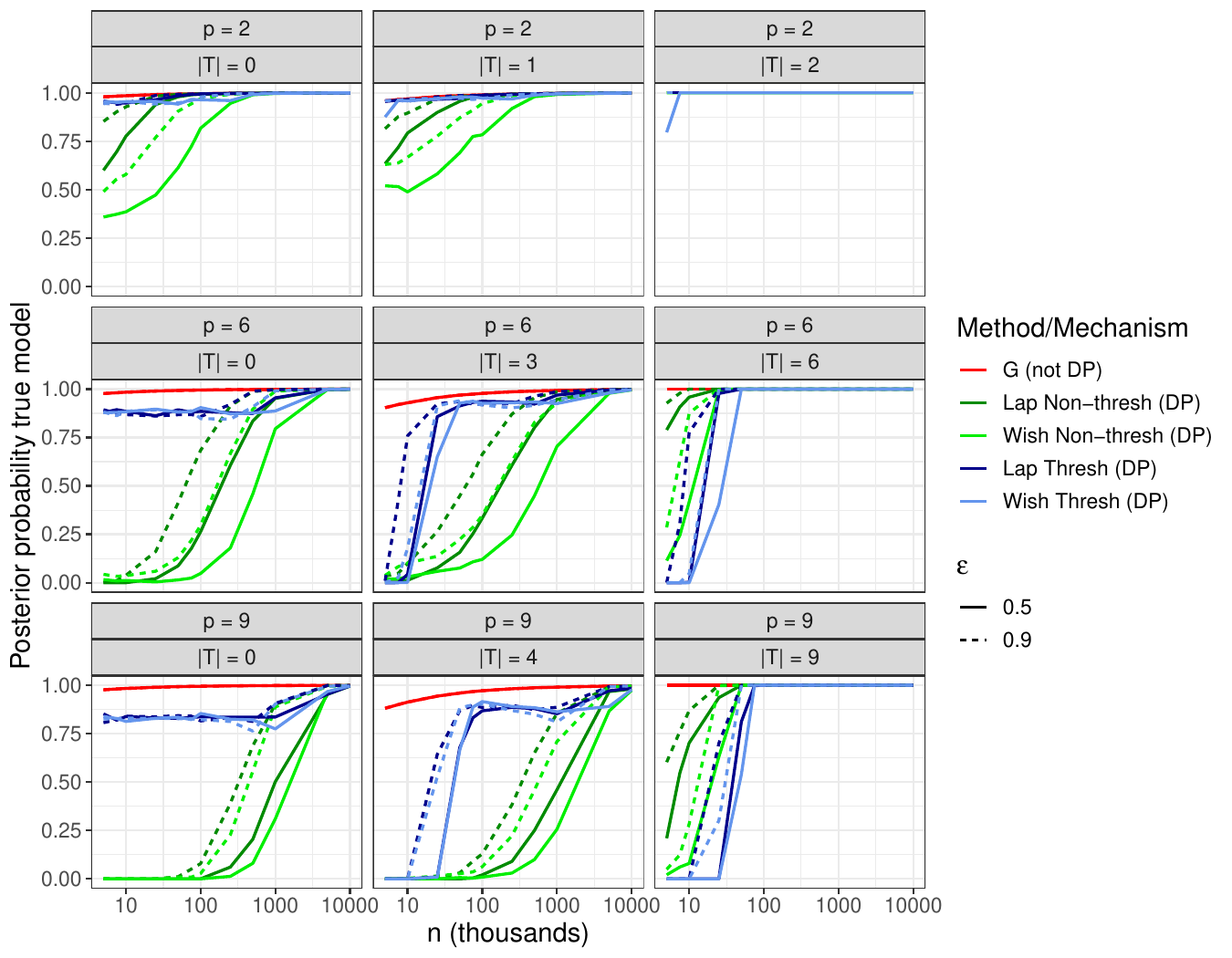}
 \caption{Simulation study: Sample size ($x$-axis) against posterior probability of the true model ($y$-axis) with BIC prior. }
\label{suppl:fig:Inclusion_Prob_BIC}
\end{figure}

\end{appendix}

\end{document}